\documentclass[letterpaper,12pt]{article}
\usepackage[utf8]{inputenc}
\usepackage[top=3cm, bottom=3cm, left=3cm, right=3cm]{geometry}
\usepackage{graphicx,rotating,amsthm,amsbsy,amssymb,amsmath, hyperref, natbib, bbm, color}
\usepackage{tikz}
\usetikzlibrary{trees,automata,positioning}

\usepackage{caption}
\usepackage{subcaption}

\renewcommand{\P}{\mathbbm{P}}
\newcommand{\1}{\mathbbm{1}}

\newtheorem{definition}{Definition}
\newtheorem{proposition}{Proposition}
\newtheorem{example}{Example}

\title{Differentially Private Hypothesis Testing \\
with the Subsampled and Aggregated  \\ Randomized Response Mechanism}
\author{Víctor Peña, Andrés F. Barrientos \\ 
\textit{Universitat Polit\`ecnica de Catalunya, Florida State University}}

\begin{document}

\maketitle

\begin{abstract}
 Randomized response is one of the oldest and most well-known methods for analyzing confidential data. However, its utility for differentially private hypothesis testing is limited because it cannot achieve high privacy levels and low type I error rates simultaneously. In this article, we show how to overcome this issue with the subsample and aggregate technique. The result is a general-purpose method that can be used for both frequentist and Bayesian testing. {{We illustrate the performance of our proposal in three scenarios: goodness-of-fit testing for linear regression models, nonparametric testing of a location parameter with the Wilcoxon test, and the nonparametric Kruskal-Wallis test. }}
\end{abstract}

\newpage

\section{Introduction}

In this article, we propose a method for testing hypotheses with confidential data. It is conceptually simple, widely applicable, and can attain high privacy levels and low type I error rates at the same time. 

We work within the \textit{differential privacy} framework \citep{dwork2006calibrating}. From a data privacy perspective, differentially private algorithms are appealing because they are robust to deanonymization attacks \citep{dwork2014algorithmic}. From a statistical perspective, differentially private algorithms are useful because they facilitate making inferences from private data.

There is a growing literature on differentially private hypothesis testing. For example, \cite{gaboardi2016differentially} and \cite{rogers2017new} provide differentially private chi-squared tests,  \cite{couch2019differentially} develop differentially private versions of nonparametric tests such as the Mann-Whitney and Kruskal-Wallis tests, and  \cite{barrientos2019differentially}, \cite{pena2021differentially} and \cite{alabi2022hypothesis} propose methods for testing in linear regression models. 

Our proposal is applying the subsample and aggregate technique \citep{nissim2007smooth} to randomized response \citep{warner1965randomized}. The result is a general-purpose algorithm that can create differentially private versions of practically any existing nonprivate hypothesis test. {{Through simulation studies and an application, we find that the method is especially useful when the type I error $\alpha$ of the tests is low. Testing hypothesis with low significance levels (as low as $\alpha = 0.005$) has been proposed as a way to ameliorate what has become known as the \textit{replication crisis}, where published significant results (typically at significance level $\alpha = 0.05$) fail to replicate in subsequent follow-up experiments \citep{benjamin2018redefine}.}}

The subsample and aggregate technique consists in splitting the data into subsets, computing statistics within them, and combining the results in a way that ensures that the output is differentially private. From a theoretical perspective, \cite{smith2011privacy} studies general asymptotic properties of the strategy. From an applied perspective, the subsample and aggregate technique has been used to build differentially private algorithms for clustering \citep{mohan2012gupt, su2016differentially}, feature selection with the LASSO \citep{thakurta2013differentially}, hypothesis testing for normal linear models \citep{barrientos2019differentially, pena2021differentially}, and logistic regression \citep{mohan2012gupt}.

Randomized response was originally motivated as a method for reducing bias in answers to sensitive questions. Since its inception more than fifty years ago, it has been extended and applied to many different contexts; the reader is referred to \cite{blair2015design} or the monograph \cite{chaudhuri2020randomized} for further details. Importantly, randomized response is differentially private \citep{dwork2014algorithmic}. Its properties within the framework have been studied in \cite{wang2016using} and \cite{ma2021randomized}, and it has been used as a building block for differentially private algorithms in \cite{erlingsson2014rappor}, \cite{bassily2015local}, and \cite{ye2019privkv}. 

Unfortunately, randomized response by itself is not useful for differentially private hypothesis testing. As we argue in Section~\ref{sec:prelim}, it cannot achieve acceptable privacy levels and type I error rates simultaneously. Fortunately, this issue can be resolved with the subsample and aggregate technique.

The output of our method is a binary decision that indicates whether we reject a null hypothesis or not. The decision can be used for both frequentist and Bayesian hypothesis testing. For the latter, one has to specify prior probabilities on the hypotheses and a prior distribution on the power of the nonprivate test used for building the private test. 

{{Previous work on differentially private hypothesis testing has focused on releasing differentially private $p$-values or, from a Bayesian perspective, Bayes factors. In contrast, our output is a binary decision. While our output is, in some sense, less informative, it need not be less useful from a practical standpoint. If we perform a hypothesis test, the type I error $\alpha$ must be set in advance. If we are using a $p$-value to make that decision, we should reject the null hypothesis if it is less than $\alpha$. Otherwise, we can be tempted to $p$-hack \citep{gelman2013garden} or misinterpret the $p$-value \citep{schervish1996p}. In our application and simulation studies in Section~\ref{sec:illustrations}, we see that approaches based on binarized outcomes are often more powerful than an approach based on $p$-values. This makes intuitive sense, since a bit (a decision to reject or not reject a null hypothesis) is less informative than a $p$-value.}}

In Section~\ref{sec:prelim}, we define differential privacy and randomized response. In Section~\ref{sec:SARR}, we define the subsampled and aggregated randomized response mechanism, study its  properties, and devise simple strategies to implement it in practice. In Section~\ref{sec:illustrations}, we illustrate the performance of the method in differentially private implementations of the goodness-of-fit tests proposed in \cite{pena2006global}, the one-sample Wilcoxon test, and the Kruskal-Wallis test.  Section~\ref{sec:discussion} closes the article with a brief discussion and ideas for future work. All proofs are relegated to the Supplementary Material to the article. {{The Supplementary Material also includes an additional simulation study comparing our method to the differentially private test for regression coefficients proposed in \cite{barrientos2019differentially}.}}

\section{Preliminaries} \label{sec:prelim}

In this section, we give a brief introduction to differential privacy and randomized response. We state a simplified version of the general definition of differential privacy that is sufficient for our purposes.

Before we define differential privacy, we need to define what neighboring datasets are first.  

\begin{definition}
Let $D = (x_1, x_2, \, ... \, , x_n) \in \{0,1\}^n$ and  $D' = (x'_1, x'_2, \, ... \, , x'_n) \in \{0,1\}^n$. Then, $D$ and $D'$ are neighbors if they differ in only one component: $x_i = x_i'$ for all $i \in \{1,2, \, ... \, , n\}$ except for one $j \in \{1,2, \, ... \, , n\}$ for which $x_j \neq x'_j$.
\end{definition}

Differential privacy bounds the extent to which the output of randomized algorithms can vary for neighboring datasets. In the differential privacy literature, privacy-ensuring randomized algorithms are referred to as \textit{mechanisms}. We formally define differential privacy below.

\begin{definition} \label{def:DP} A mechanism ${M}: \{0,1\}^n \rightarrow \{0,1\}$ is $\varepsilon$-differentially private  if there exists $\varepsilon > 0$ such that for all neighboring $D, D' \in \{0,1\}^n$ 
$$
\max\left\{ \frac{\P[{M}(D) = 1]}{\P[{M}(D') = 1]} , \frac{\P[{M}(D') = 1]}{\P[{M}(D) = 1]}, \frac{\P[{M}(D) = 0]}{\P[{M}(D') = 0]}, \frac{\P[{M}(D') = 0]}{\P[{M}(D) = 0]}  \right\} \le e^{\varepsilon}.
$$
The mechanism is exactly $\varepsilon$-differentially private if the upper bound is tight.
\end{definition}

Low values of $\varepsilon$ are associated with high privacy levels, whereas high values of $\varepsilon$ are associated with low privacy. Figure~\ref{fig:intro} illustrates how $\varepsilon$ restricts $\P[{M}(D) = 1]$ and $\P[{M}(D') = 1]$ for $\varepsilon \in \{0.1, 0.5, 1\}$. As $\varepsilon$ goes to zero,  $\P[{M}(D) = 1]$ and  $\P[{M}(D') = 1]$ are forced to be equal; as $\varepsilon$ goes to infinity, any values of $\P[{M}(D) = 1]$ and $\P[{M}(D') = 1]$ satisfy $\varepsilon$ differential privacy.

\begin{figure}
\includegraphics[width=\linewidth]{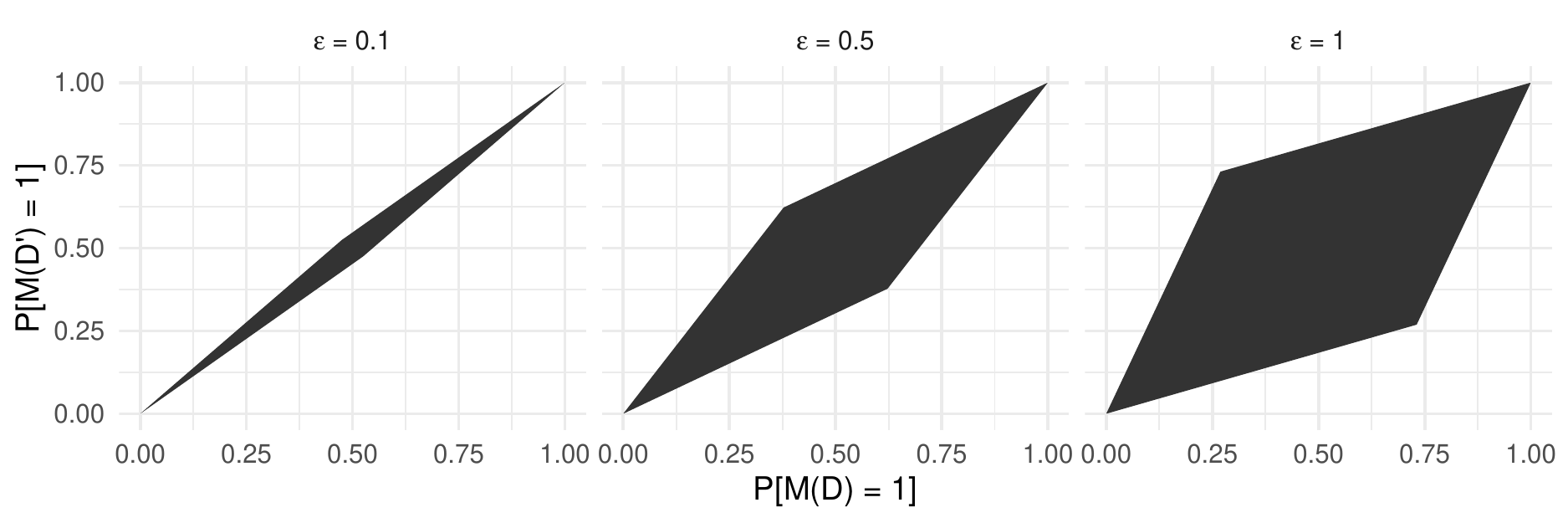}
\caption{Shaded regions show the choices of $\P[{M}(D) = 1]$ and $\P[{M}(D') = 1]$ that achieve $\varepsilon$ differential privacy for $\varepsilon \in \{0.1, 0.5, 1\}$. }
\label{fig:intro}
\end{figure}

A key building block of our method is the randomized response mechanism. It takes a binary input $x \in \{0,1\}$ and outputs
$$
r(x) = \begin{cases}
x, \text{ with probability } p, \\
1-x, \text{ with probability } 1-p, \qquad 1/2 < p < 1.
\end{cases}
$$
In this article, we think of $x$ as the outcome of a hypothesis test, and $x = 1$ implies rejection of a null hypothesis $H_0$.

The proposition below is a well-known fact in the differential privacy literature (see, for example, \cite{dwork2014algorithmic}) and states that $r(x)$ is $\varepsilon$-differentially private.

\begin{proposition}
Let $\varepsilon > 0$ and $p = \exp(\varepsilon)/[1+\exp(\varepsilon)]$. Then, $r(x)$ is exactly $\varepsilon$-differentially private. 
\end{proposition}

We would like $r(x)$ to be $\varepsilon$-differentially private and have a type I error rate of at most $\alpha$. Unfortunately, $r(x)$ cannot achieve low values of $\varepsilon$ and $\alpha$ simultaneously. If $x$ is conducted at significance level $0 < \alpha_0 < 1$, the type I error of $r(x)$ is $p \alpha_0 + (1-p)(1-\alpha_0) \ge 1-p$, which is very limiting; for example, if $\varepsilon = 1$, the type I error of $r(x)$ is at least $0.268.$ 

We could control the type I error of $r(x)$ by randomizing it further. That is, we could report $B r(x)$ for  $B \sim \mathrm{Bernoulli}(\varrho)$, where $\varrho$ is set so that $B r(x)$ has type I error $\alpha$. However, the introduction of $B$ comes at the cost of a substantial loss in power. In particular, the power of $B r(x)$ is bounded above by $\varrho$: for instance, if $\varepsilon = 1$, $\alpha_0 = 0.05$, and $\varrho$ is so that the type I error of $r(x)$ is $\alpha = 0.05$, the power of $B r(x)$ is bounded above by $\varrho \approx 0.17$. In Sections~\ref{sec:SARR} and~\ref{sec:illustrations}, we show that subsampling and aggregating provides a more powerful solution.

\section{Subsampled and aggregated randomized response} \label{sec:SARR}

\subsection{General properties}

In this section, we define the subsampled and aggregated randomized response mechanism and study its properties. 

First, we split the data uniformly at random into $2k+1$ disjoint subsets indexed by $i \in \{1, 2, \, ... \, , 2k+1\}$, where $k$ is a nonnegative integer. Within the subsets, we run the nonprivate test of interest at significance level $\alpha_0$. The outcomes of the tests are denoted $x_i$, where $x_i = 1$ indicates rejection of the null hypothesis $H_0$ in the $i-$th subset. Then, we apply independent randomized response mechanisms on the $x_i$, obtaining $r(x_i)$. Finally, we combine the results in $T = \sum_{i =1}^{2k+1} r(x_i)$ and report $d_c = \mathbbm{1}(T > c)$.

The proposition below shows that the privacy level $\varepsilon$ of $d_c$ has a closed-form expression. We derived it using facts about stochastically ordered random variables found in \cite{shaked2007stochastic}. We use the notation $\mathrm{Binomial}(i,p) + \mathrm{Binomial}(j,q)$ for the distribution of the sum of independent $\mathrm{Binomial}(i,p)$ and $\mathrm{Binomial}(j,q)$ random variables, with the understanding that if the number of trials is zero, the random variable is zero with probability one.

 \begin{proposition} \label{prop:epsilon}
The statistic $d_c = \mathbbm{1}( T > c)$ is exactly $\varepsilon$-differentially private with 
$$
\varepsilon =   \log \left( \frac{\P(B_1 > c_\ast) }{\P(B_0 > c_\ast)} \right),
$$
where $c_\ast = \max(c, 2k-c)$ and  $B_i \sim \mathrm{Binomial}(i,p)+ \mathrm{Binomial}(2k+1-i, 1-p)$ for $i \in \{0,1\}$.
\end{proposition}

Proposition~\ref{prop:epsilon} shows that $\varepsilon$ depends on $c$, $k$, and $p$. We study how these parameters affect $\varepsilon$ by fixing two of them at a time and letting the other one vary. 

\begin{proposition} \label{prop:epstk}
The statistic $d_c = \1(T > c)$ has the following properties:
\begin{enumerate}
    \item For any fixed $k$ and $c$, $\varepsilon$ is increasing in $p$.
    \item For any fixed $p$ and $c \ge k$, $\varepsilon$ is decreasing in $k$. 
    \item For any fixed $k$ and $p$, $\varepsilon$ is minimized at $c = k$. 
\end{enumerate}
\end{proposition}

Proposition~\ref{prop:epstk} establishes that subsampling and aggregating lowers $\varepsilon$ whenever $c \ge k$. It also suggests the majority vote $d = \mathbbm{1}(T > k)$ as a default choice of $d_c$, since it minimizes $\varepsilon$ for any fixed $k$ and $p$. For this reason and its intuitive appeal, we restrict our attention to $d$ from this point onward.

\begin{figure}
\begin{subfigure}{.5\textwidth}
  \centering
  \includegraphics[width=\linewidth]{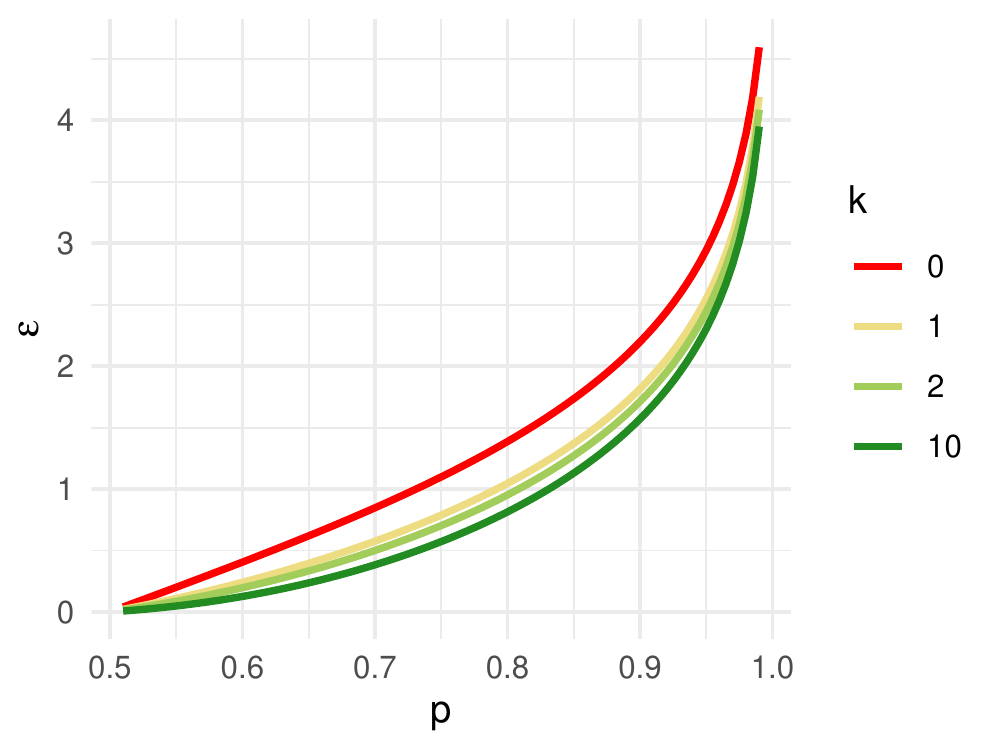}
  \caption{}
  \label{fig:eps}
\end{subfigure}%
\begin{subfigure}{.5\textwidth}
  \centering
  \includegraphics[width=\linewidth]{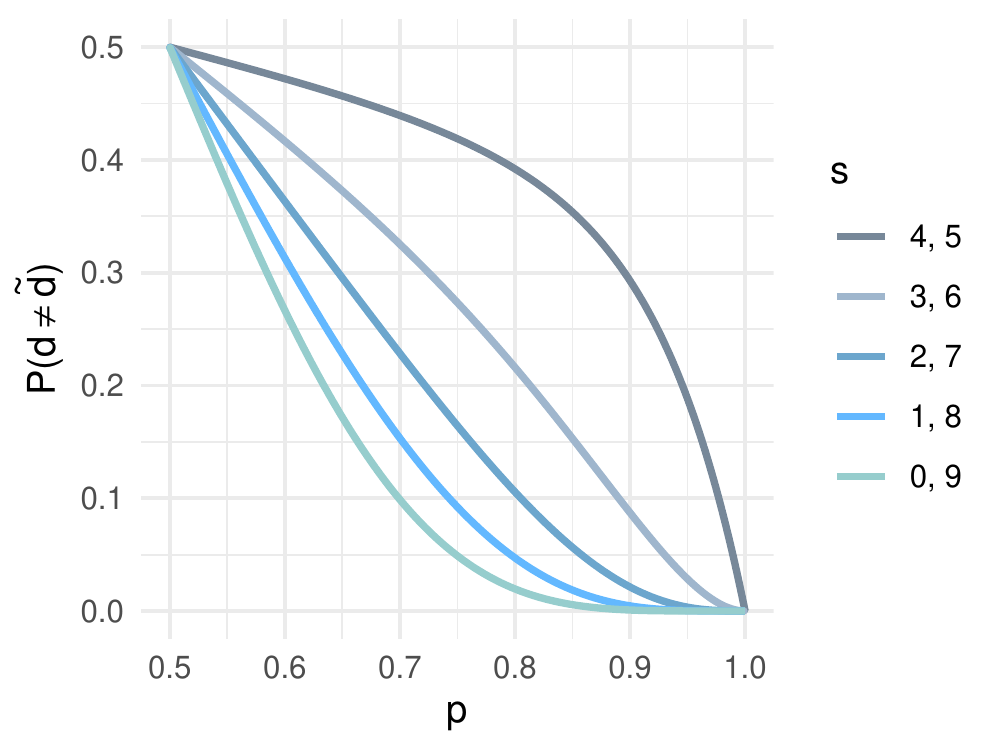}
  \caption{}
  \label{fig:disagree}
\end{subfigure}
\caption{(a) Privacy parameter $\varepsilon$ as a function of $p$ and $k$. (b) $\P(d \neq \tilde{d})$ as a function of $p$ and $T = s$ for $k = 4$.}
\label{fig:fig}
\end{figure}

Figure~\ref{fig:eps} shows $\varepsilon$ given $k$ and $p$. Splitting the data into more and more subsets reduces $\varepsilon$, but the figure seems to indicate that the gains are limited as $k$ increases. The proposition below confirms this intuition: the limit of $\varepsilon$ as $k$ goes to infinity is a positive constant that is increasing in $p$. Combined with Proposition~\ref{prop:epstk}, the limit of $\varepsilon$ as $k$ goes to infinity establishes a nontrivial necessary condition on $p$ for achieving $\varepsilon$ differential privacy.

\begin{proposition} \label{prop:properties}
The statistic $d = \1(T > k)$ has the following properties:
\begin{enumerate}
    \item For any fixed $p$, 
$$
\lim_{k \rightarrow \infty} \varepsilon  = \log\left(1+\frac{(2p-1)^2}{2p(1-p)}\right) > 0.
$$
\item A necessary condition on $p$ for achieving $\varepsilon$ differential privacy is
$$
p \le \frac{1}{2} \left( 1 + \frac{\sqrt{\exp(2 \varepsilon)-1}}{1+\exp(\varepsilon)} \right).
$$
A sufficient condition on $p$ for achieving $\varepsilon$ differential privacy is 
$$
p \le \frac{\exp(\varepsilon)}{1+\exp(\varepsilon)}.
$$
\end{enumerate}
\end{proposition}

There are two sources of uncertainty in $d$: the uncertainty in the $x_i$ and the uncertainty introduced by the randomized response mechanisms. We focus on the latter now, comparing $d = \mathbbm{1}(T > k)$ to
$\tilde{d} = \mathbbm{1}(\sum_{i=1}^{2k+1} x_i > k)$ treating $\sum_{i=1}^{2k+1} x_i$ as fixed.

Given $\sum_{i=1}^{2k+1} x_i = s$, the probability that $d$ is not equal to $\tilde{d}$ is
$$
\mathbbm{P}( d \neq \tilde{d} \mid \textstyle \sum_{i=1}^{2k+1} x_i = s ) = \begin{cases} 
\mathbbm{P}(B_s > k), \text{ if } s \le k, \\
\mathbbm{P}(B_s \le k), \text{ if } s > k,
\end{cases}
$$
where $B_s \sim \mathrm{Binomial}(s, p) + \mathrm{Binomial}(2k +1 - s, 1-p)$. Figure~\ref{fig:disagree} displays this probability as a function of $p$ and $s$ for $k = 4$. There is an interesting symmetry in $s$ that holds in general.

\begin{proposition} \label{prop:symmetry}
For any given $s \in \{0, 1, \, ... \, , k\}$, 
$$\P(d \neq \tilde{d} \mid \textstyle \sum_{i=1}^{2k+1} x_i =  s) = \P(d \neq \tilde{d} \mid \sum_{i=1}^{2k+1} x_i = 2k+1-s).$$ 
\end{proposition}

The probability that $d$ and $\tilde{d}$ disagree depends on $s, k$, and $p$. Below, we describe this dependence.

\begin{proposition} \label{prop:disagree}
The probability $\mathbbm{P}( d \neq \tilde{d} \mid \textstyle \sum_{i=1}^{2k+1} x_i = s )$ has the following properties: \label{prop:disagree} 
\begin{enumerate}
    \item For any fixed $k$ and $p$, $ \mathbbm{P}( d \neq \tilde{d} \mid \textstyle \sum_{i=1}^{2k+1} x_i = s )  $ is  decreasing in $s$ if $s > k$ and increasing in $s$ if $s \le k$.
    \item For any fixed $p$ and $s$, $\mathbbm{P}( d \neq \tilde{d} \mid \textstyle \sum_{i=1}^{2k+1} x_i = s ) $ is {decreasing} in $k$ if $s \le k$ and increasing in $k$ if $s > k$.
    \item For any fixed $k$ and $s$,  $\mathbbm{P}( d \neq \tilde{d} \mid \textstyle \sum_{i=1}^{2k+1} x_i = s )  = 1/2$ if $p = 1/2$ and $\mathbbm{P}( d \neq \tilde{d} \mid \textstyle \sum_{i=1}^{2k+1} x_i = s )  = 0$ if $p = 1$.
\end{enumerate}
\end{proposition}

A direct consequence of Propositions~\ref{prop:symmetry} and \ref{prop:disagree} is that the probability that $d$ and $\tilde{d}$ disagree is minimized when $\textstyle \sum_{i=1}^{2k+1} x_i \in \{0, 2k+1\}$ and maximized when $\textstyle \sum_{i=1}^{2k+1} x_i \in \{k, k+1\}$.

\subsection{Hypothesis testing}
\label{sec:hyptest}

In this section, we focus on properties related to hypothesis testing. For simplicity, we assume that the subsets are balanced (i.e., they have the same sample size), but the method can be used provided that the subsets are not heavily unbalanced.  Throughout, we assume that the tests behind the $x_i$ are all conducted at a fixed significance level $\alpha_0$.

The type I error and power of $d$ depend on the probability that the nonprivate test $x_i$ rejects $H_0$, which we denote $\gamma_0$. Under $H_0$, $\gamma_0$ is the type I error $\alpha_0 = \P(x_i = 1 \mid H_0)$; under $H_1$, it is the power $\P(x_i = 1 \mid H_1)$. 

Let $T = \sum_{i=1}^{2k+1} r(x_i) \sim \mathrm{Binomial}(2k+1,  p \gamma_0 + (1-p) (1-\gamma_0))$ be the number of subsets where $H_0$ is rejected. Since $d = \1(T > k)$, the distribution of $T$ lets us quantify how the probability of rejecting $H_0$ depends on $k$, $p$, and $\gamma_0$. 

\begin{proposition} \label{prop:fpr} The probability that $d$ rejects $H_0$ has the following properties:  
\begin{enumerate}
    \item For any fixed $k$ and $p$,  the probability that $d$ rejects $H_0$ is increasing in $\gamma_0$.
    \item For any fixed $\gamma_0$ and $k$, the probability that $d$ rejects $H_0$ is decreasing in $p$ if $\gamma_0 < 1/2$ and increasing in $p$ if $\gamma_0 > 1/2$.
    \item Let $p > 1/2$ be fixed. If $\gamma_0 > 1/2$, then the probability that $d$ rejects $H_0$ goes to 1 as $k \rightarrow \infty$. Alternatively, if $\gamma_0 < 1/2$, then the probability that $d$ rejects $H_0$ goes to 0 as $k \rightarrow \infty$.
\end{enumerate}
\end{proposition}

Part 3 of Proposition~\ref{prop:fpr} establishes that $d$ is consistent under $H_1$ as long as the power of the tests within the subsets is greater than $1/2$ as $k$ goes to infinity.
 
The type I error $\alpha$ of $d$ depends on $k$, $p$, and $\alpha_0$. By Proposition~\ref{def:DP}, $\varepsilon$ does not depend on $\alpha_0$, so lowering $\alpha_0$ decreases $\alpha$ without sacrificing $\varepsilon$.  In Proposition~\ref{prop:properties}, we saw that $\varepsilon$ is decreasing in $k$, but the gains are limited. This is not the case for $\alpha$, in the sense that the minimum $\alpha$ attainable as $k$ grows to infinity is zero. However, we should bear in mind that reducing $\alpha_0$ will decrease power.

\begin{proposition} \label{prop:minalpha}
For any $\varepsilon > 0$, the minimum type I error $\alpha$ attainable by $d$ goes to zero as $k$ goes to infinity.
\end{proposition}


\subsection{Tuning parameters of the mechanism} \label{sec:tuning}

\begin{figure}
\begin{subfigure}{.5\textwidth}
  \centering
  \includegraphics[width=\linewidth]{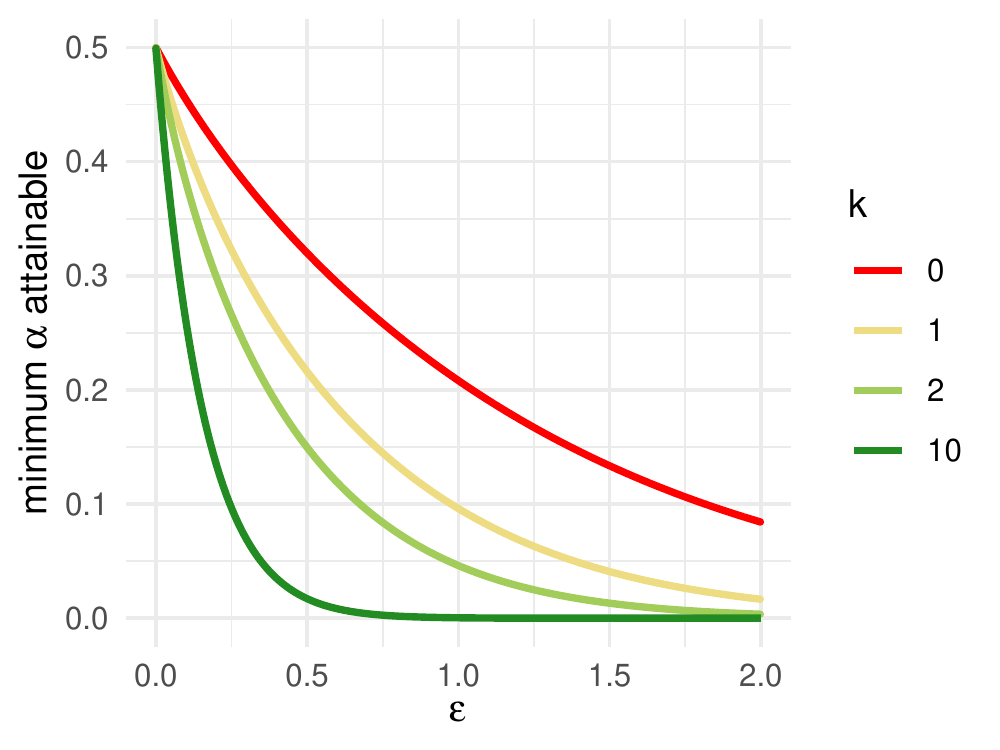}
  \caption{}
  \label{fig:minalpha}
\end{subfigure}%
\begin{subfigure}{.5\textwidth}
  \centering
  \includegraphics[width=\linewidth]{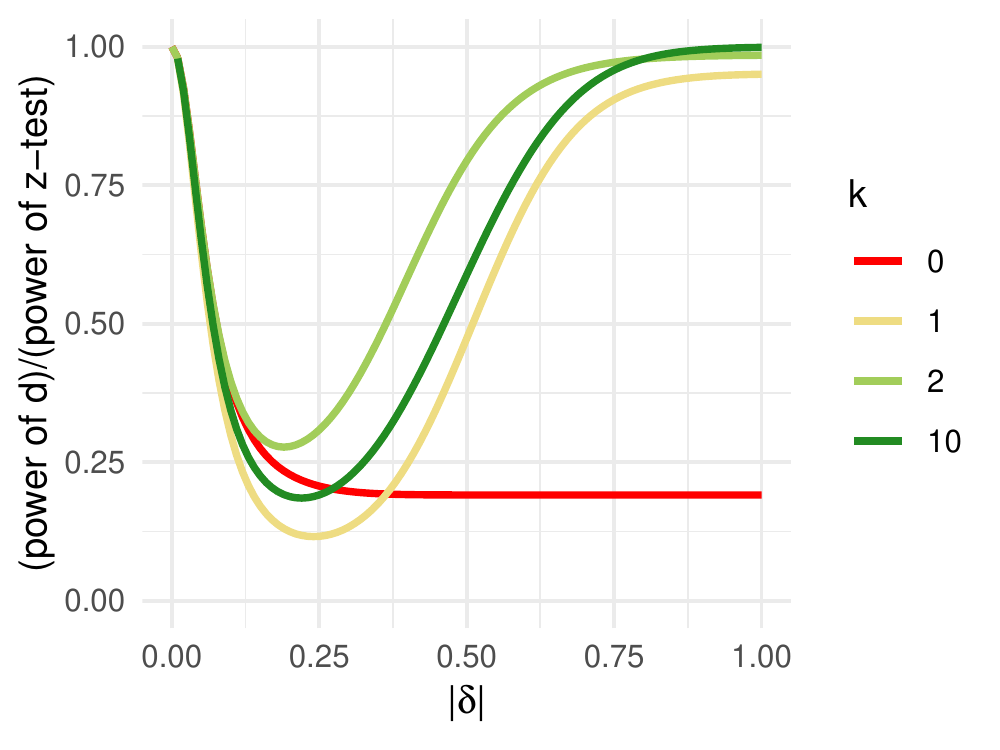}
  \caption{}
  \label{fig:powztest}
\end{subfigure}
\caption{(a) Minimum $\alpha$ attainable as a function of $\varepsilon$ and 
$k$. (b)  Difference in power between $d$ and $z$-test as a function of effect size $|\delta|$ and $k$ for $n = 105$.}
\label{fig:fig}
\end{figure}

We propose two strategies for choosing $k$, $p$, and $\alpha_0$. In both cases, we set $p$ given $k$ so that $d$ is exactly $\varepsilon$-differentially private. Once $k$ and $p$ are fixed, we find an $\alpha_0$ such that $d$ has type I error $\alpha$.

The first strategy is inspecting power curves for values of $k$ on a grid. For each $k$, we determine if there are $p$ and $\alpha_0$ such that $d$ is $\varepsilon$-differentially private and has type I error $\alpha$. If such $p$ and $\alpha_0$ exist, we can find a power curve. In some cases, power curves have closed-form expressions; in others, they can be simulated. After finding power curves for all the values of $k$ on the grid, we can choose a value visually. 

The second strategy is a heuristic tailored to the scenario where the tests are low-powered. When the tests are low-powered, increasing the number of subsets typically decreases the power of $d$. With that in mind, we propose setting $k$ to the minimum value $k^\ast$ for which there exist $p$ and $\alpha_0$ such that $d$ is $\varepsilon$-differentially private and has type I error $\alpha$. To avoid $\alpha_0$ that are too small, we add the restriction $\alpha_0 \ge \alpha_{0, \min}$ for a user-defined $\alpha_{0, \min} > 0$.

\begin{table}[h]
\centering
\caption{Minimum $k$ needed for different combinations of $\alpha$ and $\varepsilon$.}
\begin{tabular}{r|r|r|r|r|r}
\hline
  min $k$  & {$\varepsilon = 0.5$} & $\varepsilon = 0.75$ &  $\varepsilon = 1$ &  $\varepsilon = 1.25$ & $\varepsilon = 1.5$ \\
\hline
$ \alpha = 0.005$ & 13 & 8 & 6 & 4 & 3\\
\hline
$\alpha = 0.01$ & 11 & 7 & 5 & 4 & 3\\
\hline
$\alpha = 0.05$ & 6 & 4 & 3 & 2 & 1\\
\hline
$\alpha = 0.1$ & 4 & 2 & 2 & 1 & 1\\
\hline
\end{tabular}
\label{tab:minalpha}
\end{table}

{{Proposition~\ref{prop:minalpha}} guarantees that, for $k$ large enough, we can achieve any privacy level $\varepsilon$ and type I error $\alpha$ simultaneously. However, if we want both $\varepsilon$ and $\alpha$ to be small and we have a small sample size, we may not be able to split the data into enough subsets to satisfy both requirements. Table~\ref{tab:minalpha} shows the minimum $k$ needed to simultaneously achieve type I error $\alpha$ and $\varepsilon$ differential privacy for $\alpha \in \{0.005, 0.01, 0.05, 0.1\}$ and $\varepsilon \in \{0.5, 0.75, 1, 1.25, 1.5\}$. The lower $\alpha$ and $\varepsilon$ are, the larger $k$ needs to be to simultaneously achieve $\alpha$ type I error level and $\varepsilon$ differential privacy.}

Below, we apply the two strategies we described to the one sample $z$-test. The example illustrates the need for setting $\alpha_{0, \min} > 0$: without a nonzero minimum, $\alpha_{0, \min}$ can be too small and the mechanism can be underpowered.

\begin{example} [One sample $z$-test] \label{ex:ztest} Let the data be 105 independent and identically distributed observations distributed as $\mathrm{Normal}(\mu, \sigma^2)$ with $\sigma^2$ known. We set $\varepsilon = 1.5$ and test $H_0: \mu = \mu_0$ against $H_1 \neq \mu_0$ with the one-sample $z$-test at significance level $\alpha = 0.05$. In this example, the classical randomized response mechanism ($k = 0$) cannot achieve $\varepsilon$-differential privacy and type I error $\alpha$ at the same time. To solve the issue, we let $p = \exp(\varepsilon)/[1+\exp(\varepsilon)]$ and define $B r(x)$, where $B \sim \mathrm{Bernoulli}(\varrho)$. The probability $\varrho$ is set so that $B r(x)$ has type I error $\alpha$.
Figure~\ref{fig:powztest} shows the ratio of the  power of $d$ to the power of the usual nonprivate $z$-test for $k \in \{0, 1, 2, 10\}$ and effect sizes $|\delta| = |\mu-\mu_0|/\sigma$ ranging from 0 to 1. The classical randomized response mechanism $(k = 0)$ and $k = 1$ do not perform well. In the latter case, the method has low power because $\alpha_0 \approx 0.0025$. The performance of $k = 2$ and $k = 10$ is similar for small effect sizes. For moderate effect sizes, $k = 2$ is preferable. For large effect sizes, $k = 10$ ends up outperforming $k = 2$, but at that point both approaches are essentially as powerful as the nonprivate $z$-test. The values of $\alpha_0$ are 0.089 for $k = 2$ and 0.281 for $k = 10$, respectively. If we use the automatic strategy to select the parameters of the mechanism with $\alpha_{0, \min} = 0$, it selects $k^\ast = 1$; if we set a minimum $\alpha_{0, \min} > 0.0025$, it chooses $k^\ast = 2$ instead.
\end{example}

\subsection{Extensions: multiple hypotheses and Bayesian testing}\label{sec:extensions}

The subsampled and aggregated randomized response mechanism can be used for testing multiple hypotheses. Indeed, it is straightforward to apply a Bonferroni correction to multiple independent runs of the mechanism. If each test is $\varepsilon$-differentially private, the vector $(d_1, d_2, \, ... \, , d_m)$ is $m \varepsilon$-differentially private by the sequential composition property of differential privacy \citep{mcsherry2009privacy}. So if we want to test $m$ null hypotheses $H_0^{1}, H_0^2, \, ... \, , H_0^m$ at a familywise error rate $\alpha$, we can run $m$ independent subsampled and aggregated randomized response mechanisms at significance level $\alpha/m$. We pursue this idea in Section~\ref{subsec:gof}.

The binary decision $d$ can also be used for Bayesian hypothesis testing. If $d$ is calibrated to have type I error $\alpha$, then the posterior probability of $H_1$  given that $d = 0$ is
\begin{align*}
\P(H_1 \mid d = 0)  &= \frac{\P(H_1)  \P(d = 0 \mid H_1)}{ \P(H_0) \P(d = 0 \mid H_0)  + \P(H_1) \P( d = 0 \mid H_1)} \\ &=  \frac{\P(H_1)  \P(d = 0 \mid H_1)}{ \P(H_0) (1-\alpha)  + \P(H_1) \P( d = 0 \mid H_1)}.
\end{align*}
If $\P(d = 0 \mid H_1)$ goes to zero as the sample size increases (i.e., $d$ is consistent under $H_0$), then $\P(H_1 \mid d = 0)$ goes to zero as the sample size increases for all $\alpha$ and $\P(H_0)$. Therefore, the Bayesian test can give decisive evidence in favor of $H_0$ asymptotically.

Analogously, the posterior probability of $H_1$ given that $d = 1$ is
\begin{align*}
\P(H_1 \mid d = 1) &= \frac{\P(H_1) \P( d = 1 \mid H_1)}{ \P(H_0)  \alpha + \P(H_1) \P( d = 1 \mid H_1)}.
\end{align*}
If $d$ is consistent under $H_1$, $\P(H_1 \mid d = 1)$ converges to $\P(H_1)/[ \P(H_0) \alpha + \P(H_1)].$ If the null and alternative hypotheses are equally likely \textit{a priori}, the limit simplifies to $1/(\alpha +1)$, which is greater than $1-\alpha$. In this case, the Bayesian test cannot give decisive evidence in favor of $H_1$ asymptotically, but it can give fairly strong evidence in its favor. 

For finite sample sizes, we can evaluate $\P(H_1 \mid d)$ given $\P(H_1)$, $d$, and a prior distribution on the power $\pi(\gamma_0 \mid H_1) = \P(x_i = 1 \mid H_1)$. Once this prior is set,
$$
\P(d = 1 \mid H_1) =  \int_0^1 \P(T > k \mid \gamma_0, H_1) \, \pi(\gamma_0 \mid H_1) \,  \mathrm{d} \gamma_0,
$$
where $T \mid \gamma_0, H_1 \sim  \mathrm{Binomial}(2k+1, p \gamma_0 + (1-p)(1- \gamma_0) )$.

In the absence of strong prior information about $\gamma_0 \mid H_1$, we recommend running a sensitivity analysis with different priors. When the power can be expressed as a function of an effect size $\delta$, it can be convenient to induce a prior distribution through it. We illustrate this point with the $z$-test below.

\begin{example}[Bayesian one sample $z$-test]  \label{ex:bayes_ztest}
Let the data be independent and identically distributed as $\mathrm{Normal}(\mu, \sigma^2)$. We test $H_0: \mu = \mu_0$ against $H_1: \mu \neq \mu_0$ with a Bayesian test. We set $\P(H_0) = \P(H_1) = 1/2$ and put a unit information prior $\mathrm{Normal}(\mu_0, \sigma^2)$ on $\mu \mid H_1$, which induces $\delta \mid H_1 \sim \mathrm{Normal}(0,1)$ on the effect size $\delta = (\mu-\mu_0)/\sigma$. The prior on $\delta \mid H_1$, in turn, induces a prior $\pi(\gamma_0 \mid H_1) = \P(x_i = 1 \mid H_1).$ The unit information prior is a common default choice for this problem and contains roughly as much information as one observation in the sample \citep{kass1995reference}.  We set $\varepsilon = 1.5$ and consider $k \in \{1, 2, 10\}$ and subgroup sample sizes $b \in \{2, 3, \, ... \, , 50\}$. For any given $k$ and $b$, the total sample size is $n = (2k+1) b$. The results are shown in Figure~\ref{fig:bayes}. As $b$ increases, the posterior probability is more decisive against or in favor of $H_1$, depending on whether $d = 0$ or $d = 1$, respectively. 
\end{example}

{{A peculiarity of our Bayesian analysis is that $d$ has a fixed type I error rate $\alpha$. From a strictly Bayesian perspective, one could output a binary decision $d$ that does not have a fixed type I error rate. However, we appreciate the fact that $d$ can be interpreted by both frequentists and Bayesians, which is in line with ongoing efforts for reconciling frequentist and Bayesian answers (see e.g. \cite{bayarri2004interplay} and \cite{bayarri2016rejection}).}}

{{The Bayesian approach described here is based on conditioning on a binary outcome. There are other proposals in the differential privacy literature that would approach the problem differently.}}

{{For example, an alternative approach would be to condition on perturbed sufficient statistics instead of binary outcomes, as proposed in  \cite{amitai2018differentially} and \cite{pena2021differentially}. }}

{{Another option is drawing directly from posterior distributions in a way that ensures differential privacy (see e.g. \cite{dimitrakakis2017differential}, \cite{heikkila2019differentially}, \cite{geumlek2017renyi}, and
\cite{hu2022private}).  However, these strategies often assume an upper bound on the likelihood, which may require the user to modify their models to meet the assumption. One major drawback of these methods is that they typically require a privacy budget proportional to the number of posterior samples desired. This, in turn, may lead to unreliable Monte Carlo approximations if $\varepsilon$ is small. In spite of the potential drawbacks, applying these approaches to hypothesis testing is worthy of further exploration and development.}}

{{Another promising approach would be to use the data augmentation Markov Chain Monte Carlo scheme proposed in \cite{ju2022data}, which has been applied to estimation problems, but not to hypothesis testing.}}

\begin{figure}
\includegraphics[width=\linewidth]{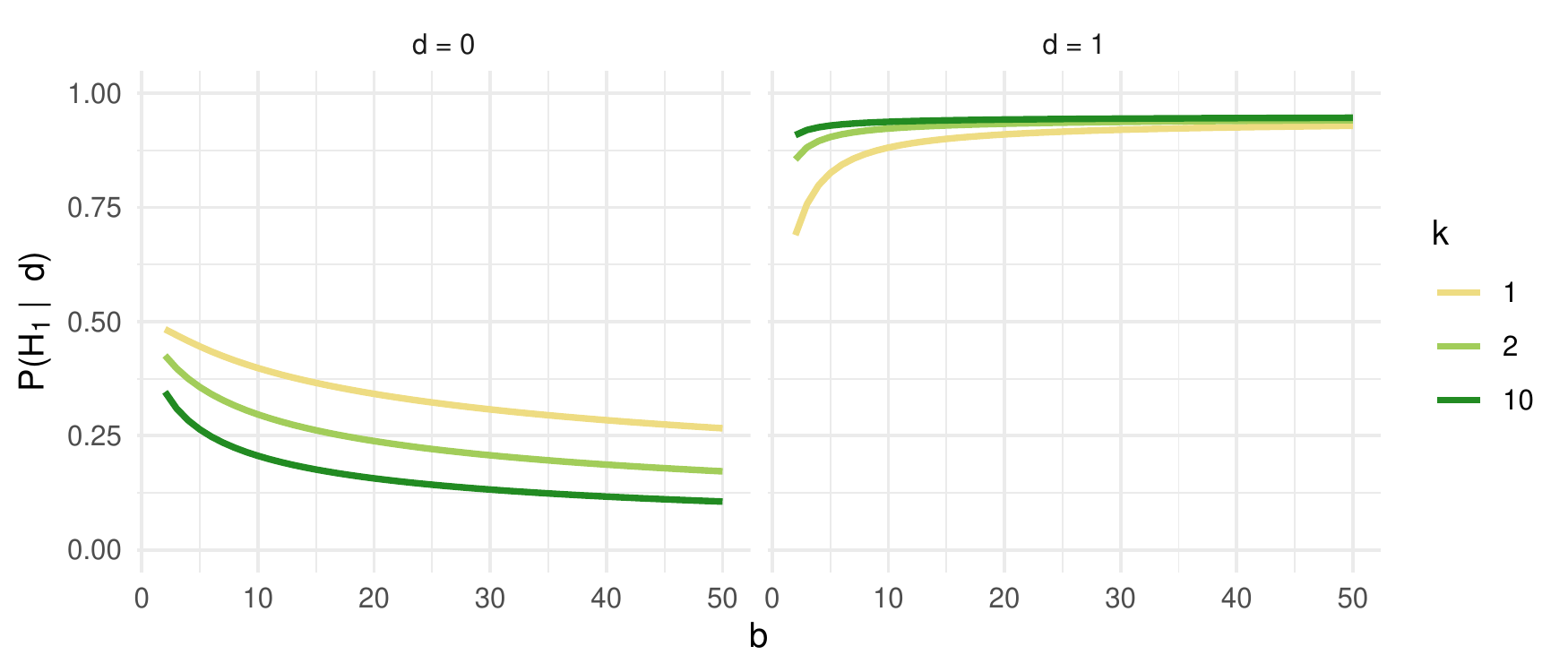}
\caption{Posterior probability $\P(H_1 \mid d)$ as a function of $k$ and the sample size within the subgroups $b$ if $\P(H_0) = 1/2$ and the effect size is $\delta \sim \mathrm{Normal}(0,1)$.}
\label{fig:bayes}
\end{figure}

\section{Illustrations}
\label{sec:illustrations}

{{In this section, we evaluate the performance of our methods in simulation studies and an application. In Section~\ref{subsec:gof}, we use the housing dataset in \cite{lei2018differentially} to implement differentially private versions of the goodness of fit tests developed in \cite{pena2006global}. In Section~\ref{subsec:bayes_wilcox}, we run a simulation study applying the Bayesian extension devised in Section~\ref{sec:extensions} to the Wilcoxon test. Finally, in Section~\ref{subsec:kw}, we compare our general-purpose method to a differentially private Kruskal-Wallis test that was proposed in \cite{couch2019differentially}. In the supplementary material, there is an additional simulation study where we compare our method to the differentially private $t$-test proposed in \cite{barrientos2019differentially}.

In our illustrations, we include the subsampled and aggregated Laplace mechanism \citep{smith2011privacy} as a competitor. For this method, we split the data into $2k+1$ subsets and run the corresponding nonprivate tests within them. The output is made differentially private after adding a Laplace perturbation term. We consider two variants of this approach.

In the first one, we find $2k+1$ $p$-values, one for each hypothesis test, and find the average $p$-value. The result is made differentially private after adding a perturbation term $\eta \sim \mathrm{Laplace}(0, 1/[\varepsilon(2k+1)])$ to the average $p$-value. The distribution of the differentially private statistic can be simulated under $H_0$, so it is straightforward to find a critical value that ensures a fixed type I error rate $\alpha$.

The second approach was suggested by an Associate Editor, and is based on the sum of binary outcomes $\sum_{i=1}^{2k+1} x_i$ instead of the average $p$-value. In that case, the output is made differentially private after adding a perturbation term $\eta \sim  \mathrm{Laplace}(0, 1/\varepsilon)$ to the sum. We can easily find a critical value that ensures a type I error rate $\alpha$ by simulating the distribution of the statistic under $H_0$. }}

\subsection{Goodness-of-fit tests for regression} \label{subsec:gof}

In this section, we study the performance of the subsampled and aggregated randomized response mechanism in a differentially private implementation of four goodness-of-fit tests for regression proposed in \cite{pena2006global}.

We perform a simulation study based on the housing dataset used in  \cite{lei2018differentially}.  The dataset contains information on houses sold in the San Francisco Bay area between 2003 and 2006. In our analysis, we only consider houses with prices within the \$105000–905000 range and sizes smaller than 3000 ft$^2$. After preprocessing, the dataset contains 235760 rows and the following variables: price (used as response $Y$), base square footage, time of transaction, lot square footage, latitude, longitude, age, number of bedrooms, and a binary variable indicating whether the house is located in a small county.

\cite{pena2006global} develop tests for checking model assumptions in the normal linear model. They provide a global test of goodness-of-fit and individual tests for detecting specific violations of assumptions. Here, we consider four tests: (1) a test whose null hypothesis is that the kurtosis of the errors is equal to 3, which is satisfied when the errors are normal, (2) a test whose null is that the errors are symmetric, (3) a test whose null is that the errors are homoscedastic, (4) and a test whose null is that the expected value of the response is linear in the predictors. For each of our simulations, we perform the four tests at significance level $\alpha/4$ and privacy level $\varepsilon/4$. This ensures that our answers have a familywise error rate $\alpha$ and a global privacy level $\varepsilon$.

To simulate data, we first fit a normal linear model using price as the response $Y$ and the remaining variables as predictors $X$, obtaining maximum likelihood estimates of the regression coefficients $\widehat{\beta}$ and the residual standard deviation $\widehat{\sigma}$. Then, we use $\widehat{\beta}$ and $\widehat{\sigma}$ along with the observed $X$ to simulate new values of the response. More precisely, we simulate data from the model $Y^\ast = X \widehat{\beta} + \widehat{\sigma} W^\ast$, where $W^\ast$ is a vector with independent and identically distributed skew-normal components with location and scale parameters equal to zero and one, respectively, and skew parameter equal to $\theta$. If $\theta = 0$, all the assumptions of the normal linear model hold, but if $\theta \neq 0$, the null hypotheses related to the kurtosis and skewness of errors are false.

We set the significance level to $\alpha \in \{0.005, 0.01, 0.05, 0.1\}$ and consider privacy parameters $\varepsilon \in \{0.5, 0.75, 1, 1.25, 1.5\}$. The skew parameter $\theta$ is ranging from 0 to 1.5. For each value of $\alpha$, $\varepsilon$ and $\theta$, we perform $10^4$ simulations. The number of subgroups $2k+1$ is determined with the automatic strategy outlined in Section~\ref{sec:tuning} with $\alpha_{0, \min} = \alpha$.

{{The results for the test of skewness can be found in Figure~\ref{fig:gof}. The results for the test of kurtosis can be found in the Supplementary Material and are similar to the ones we observe for skewness. The results for the tests of linearity and homoscedasticity are uninteresting: since the null hypothesis holds in these cases, the probability of rejecting the null is fixed to $\alpha/4$ for all $\varepsilon$ and $\theta$. In Figure~\ref{fig:gof}, we also include the ``truth'', defined as the result of running the nonprivate test without splitting the data or running any mechanisms. The subsampled-and-aggregated sum and randomized response (labeled as SARR) outperform the average $p$-value in most scenarios.  The average $p$-value is best for small $\alpha$ and $\varepsilon$, especially for low $\theta$. The performance of the sum and randomized response, which are both based on binarized outcomes, is quite similar. Randomized response performs best when $\varepsilon \in \{1, 1.25, 1.5\}$ and $\alpha = 0.005$.}}

\begin{figure}
\includegraphics[width=\linewidth]{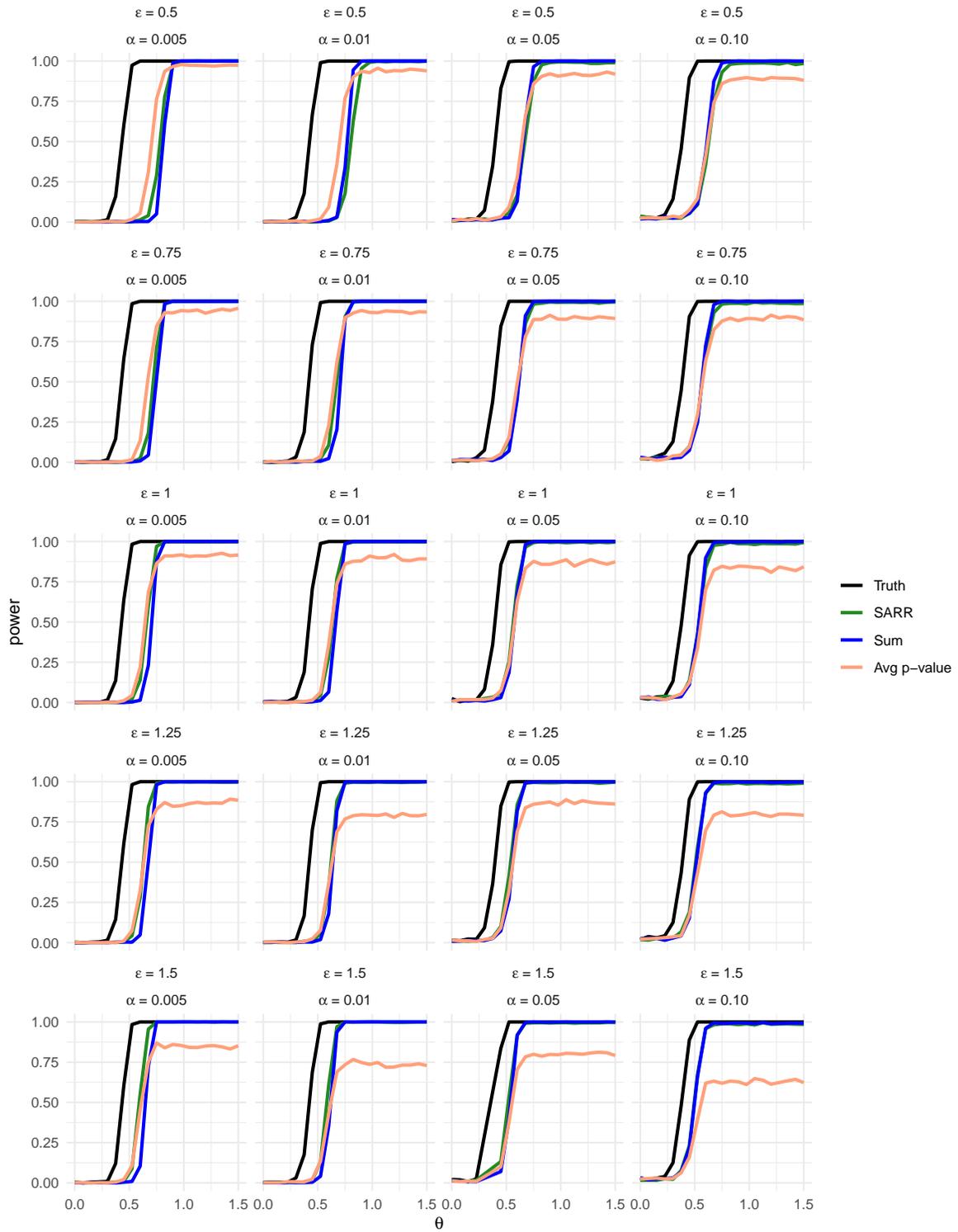}
\caption{Goodness-of-fit tests: Average power of tests for skewness for different combinations of $\alpha$ and $\varepsilon$.}
\label{fig:gof}
\end{figure}

\subsection{Bayesian answers from one-sample Wilcoxon test} \label{subsec:bayes_wilcox}

In this section, we report the results of a simulation study that compares posterior probabilities of hypotheses based on the outcomes of one-sample Wilcoxon tests, using the approach proposed in Section~\ref{sec:extensions}. 

We test $H_0: \theta =  0$ against $H_1: \theta \neq 0$ for a location parameter $\theta \in \mathbb{R}$. The tests behind the $x_i$ are one-sample Wilcoxon tests and the parameters of the mechanism are tuned so that $d$ has type I error rate $\alpha = 0.05$.

We repeatedly simulate datasets of sample size $n = 200$ comprised of independent and identically distributed observations $y_i = \theta + \tau_i$ for  $\theta \in \{0, 0.25, \, ... \, , 2\}$, where $\tau_i$ has a Student-$t$ distribution with 1.5 degrees of freedom. We consider  {{$\alpha \in \{0.005, 0.01, 0.05, 0.1\}$,  $\varepsilon \in \{0.5, 0.75, 1, 1.25, 1.5\}$}}, and run $10^4$ simulations for each combination of $\theta$, $\alpha$ and $\varepsilon$. We select $k$ with the automatic strategy we described in Section~\ref{sec:hyptest} with $\alpha_{0, \min} = \alpha$.

The prior probabilities on the hypotheses are $\P(H_0) = \P(H_1) = 1/2$. As we mentioned in Section~\ref{sec:extensions}, we need a prior on $\P(d = 1 \mid H_1)$ to perform a Bayesian test based on $d$. Table~\ref{tab:priors} lists the definitions of $\P(d = 1 \mid H_1)$ for the methods we included in the simulation study. As we did in Example~\ref{ex:bayes_ztest}, we induce a prior on $\P(d = 1 \mid H_1)$ through simpler quantities for which we can define a prior more comfortably. 

\begin{table}[t]
\centering
\caption{\label{tab:priors} Prior distributions for one-sample Wilcoxon test and their values of $\P(d = 1 \mid H_1)$. ``SA'' stands for ``subsample and aggregate.''}
\begin{tabular}{l|l}
Method       & $\P(d = 1 \mid H_1)$                          \\
\hline
Truth        & $0.8 \overline{\gamma}^z_0$  \\
SA + Randomized Response     & $\int_0^1 \P(T > k \mid \gamma_0, H_1) \, \pi(\gamma_0 \mid H_1) \,  \mathrm{d} \gamma_0$ \\
SA  + Average $p$-value & $\P( \overline{p} + \eta < c_{\alpha} \mid H_1)$    \\
SA  + Sum          &       $\P( \sum_{i=1}^{2k+1} x_i + \eta > \tilde{c}_{\alpha} \mid H_1)$  
\end{tabular}
\end{table}

For the subsampled and aggregated randomized response mechanism, we define the prior $\pi(\gamma_0 \mid H_1) = \P(x_i = 1 \mid H_1)$ as follows. Let $\overline{\gamma}^z_0$ be the average power of the $z$-test induced by the unit information prior we used in Example~\ref{ex:bayes_ztest}. Then, our prior $\pi(\gamma_0 \mid H_1)$ is a beta distribution parametrized in terms of its expected value $\mu$ and effective sample size $\kappa$ with $\mu = 0.8 \overline{\gamma}^z_0$ and $\kappa = 2k +1$ (see e.g. Chapter 6 of \cite{kruschke2014doing} for a discussion on the convenience of this parametrization in Bayesian analysis). 

{{{For the subsampled and aggregated average $p$-value}}, we put a prior on the average $p$-value $\overline{p}$ under $H_1$. First, we find the expected $p$-value for the $z$-test $\overline{p}_z$ induced by the unit information prior. Then, we define our prior on $\overline{p} \mid H_1$ as a beta distribution centered at $\mu = 0.8 \overline{p}_z$ and with an effective sample size $\kappa = 2k+1$. Given the prior  $\overline{p} \mid H_1$ and $\eta \sim \mathrm{Laplace}(0, 1/[\varepsilon(2k+1)])$, $\P(d = 1 \mid H_1)$ is the probability that $\overline{p} + \eta$ is below a critical value $c_{\alpha}$ so that  $\P(\overline{p} + \eta < c_\alpha \mid H_0) = \alpha$.

{{{For the subsampled and aggregated sum, we follow a similar strategy. We put a prior on the sum $\sum_{i =1}^{2k+1} x_i \sim \mathrm{Binomial}(2k+1, \gamma_0)$ through $\pi(\gamma_0 \mid H_1) = \P(x_i = 1 \mid H_1)$, which is a beta distribution with $\mu = 0.8 \overline{\gamma}_0^z$ and $\kappa = 2k+1$ (as it was with the subsampled and aggregated randomized response mechanism). The probability $\P(d = 1 \mid H_1)$ is the probability that $\sum_{i=1}^{2k+1} x_i + \eta$ is above a critical value $\tilde{c}_{\alpha}$ so that  $\P( \sum_{i=1}^{2k+1} x_i + \eta > \tilde{c}_\alpha \mid H_0) = \alpha$, where $\eta \sim \mathrm{Laplace}(0, 1/\varepsilon)$.}}

Lastly, we include the ``truth'', defined as the result of running the usual nonprivate Wilcoxon test, without splitting the data or running any mechanisms. For that case, we take $\P(d = 1 \mid H_1) = 0.8 \overline{\gamma}^z_0$. 

Figure~\ref{fig:wilcox} displays the results of the simulation study. {The methods that are based on binarized outcomes, namely the subsampled and aggregated randomized response mechanism and the subsampled and aggregated perturbed sum, tend to outperform the subsampled and aggregated $p$-value. The exception is the case $\varepsilon = 0.5$ and $\alpha = 0.005$. The performance of randomized response and the perturbed sum is similar. Randomized response seems to be especially helpful when $\alpha$ is low and $\varepsilon$ is larger or equal to 1.}

\begin{figure}
 \includegraphics[width=\linewidth]{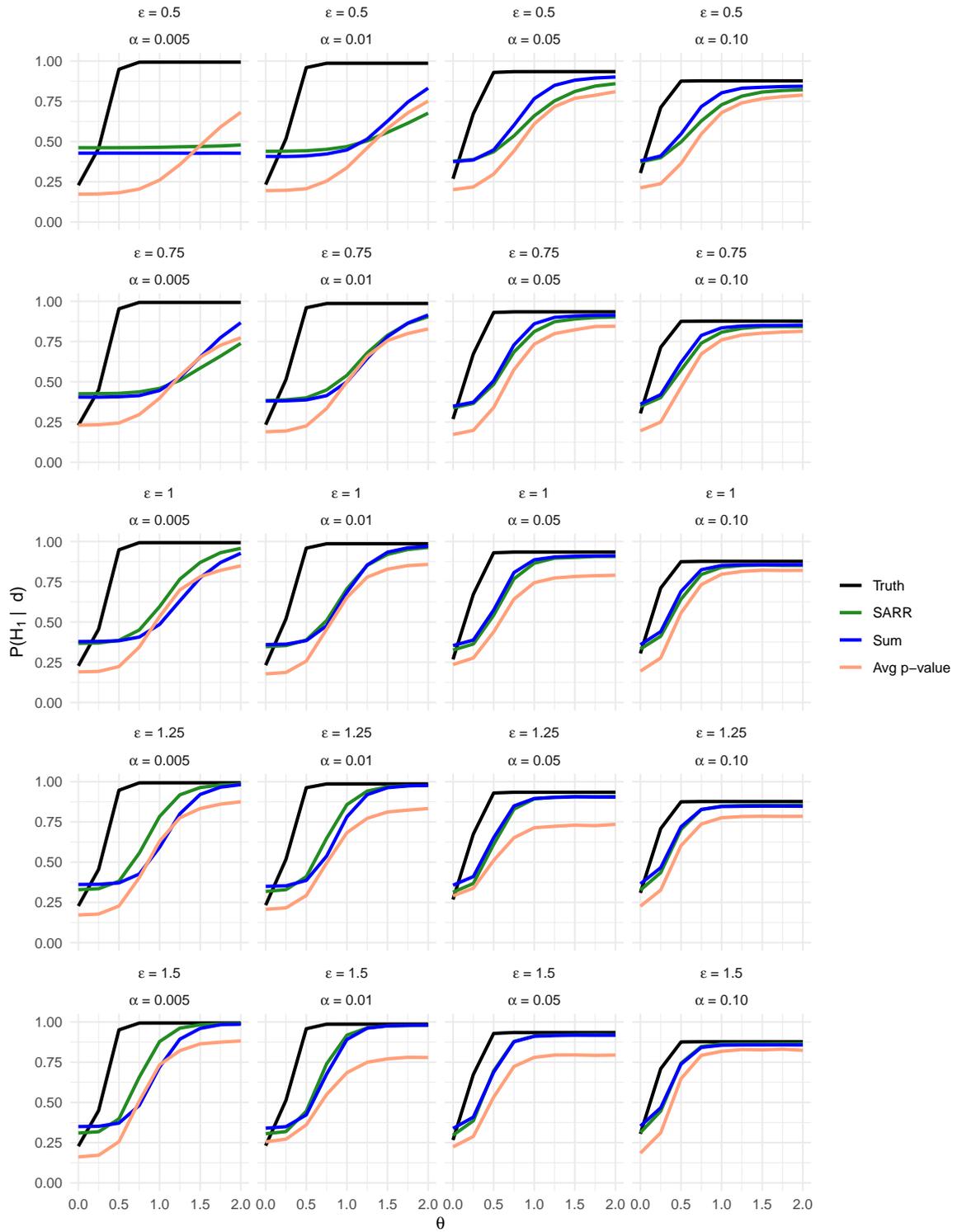}
 \caption{One-sample Wilcoxon test: Average posterior probability $\P(H_1 \mid d)$ for different values of $\alpha$, $\varepsilon$, and location parameter $\theta$.}
 \label{fig:wilcox}
 \end{figure}

{\subsection{Nonparametric ANOVA: Kruskal-Wallis test} \label{subsec:kw}}

{{In this section, we report the results of a simulation study involving the Kruskal-Wallis test. As a competitor, we include the test proposed in Section 3.3 in \cite{couch2019differentially}, which is a differentially private method that is built specifically for the Kruskal-Wallis test.}}

{{We simulate data independently from three groups: one where the data are distributed as Normal$(1,1)$, another where the data are Normal($2,1$), and another one where the data are Normal$(3,1)$. We consider sample sizes ranging from 15 to 500 in increments of 3 so that the groups are balanced. The power is approximated after performing $10^4$ simulations. This simulation study is similar to the one conducted in Section 3.4 in \cite{couch2019differentially}. We consider $\alpha \in \{0.005, 0.01, 0.05, 0.1\}$ and $\varepsilon \in \{0.5, 0.75, 1, 1.25, 1.5\}$.}}

{{As we did in the previous illustrations, we set the number of subgroups with the heuristic recommended in Section~\ref{sec:tuning}. When $n$, $\varepsilon$, and $\alpha$ are all small, the sample size is not large enough to simultaneously guarantee $\varepsilon$ differential privacy and type I error $\alpha$ with the randomized response mechanism (see Table~\ref{tab:minalpha} and the discussion in Section~\ref{sec:tuning}). For those cases, the method in \cite{couch2019differentially} can be run, but its power is low.}}

{{The results of the simulation study are displayed in Figure~\ref{fig:kruskal}. Overall, the method in \cite{couch2019differentially} (labeled as KWabs), which is tailored to this task, outperforms the general-purpose algorithms. The loss in power is most noticeable for smaller $n$, $\alpha$, and $\varepsilon$.  Randomized response is preferable over the sum when $\alpha \times \varepsilon \in \{0.005, 0.01\} \times \{1, 1.25, 1.5\}$. The sum outperforms randomized response when $\alpha \times \varepsilon \in \{0.05, 0.10\} \times \{0.5, 0.75\}$. In the remaining cases, the performance of the two approaches is relatively similar. Alternatively, the average $p$-value performs best (out of all general-purpose algorithms) for $\alpha \times \varepsilon \in \{0.005, 0.01\} \times \{0.5, 0.75\}$.}}

\begin{figure}
\includegraphics[width=\linewidth]{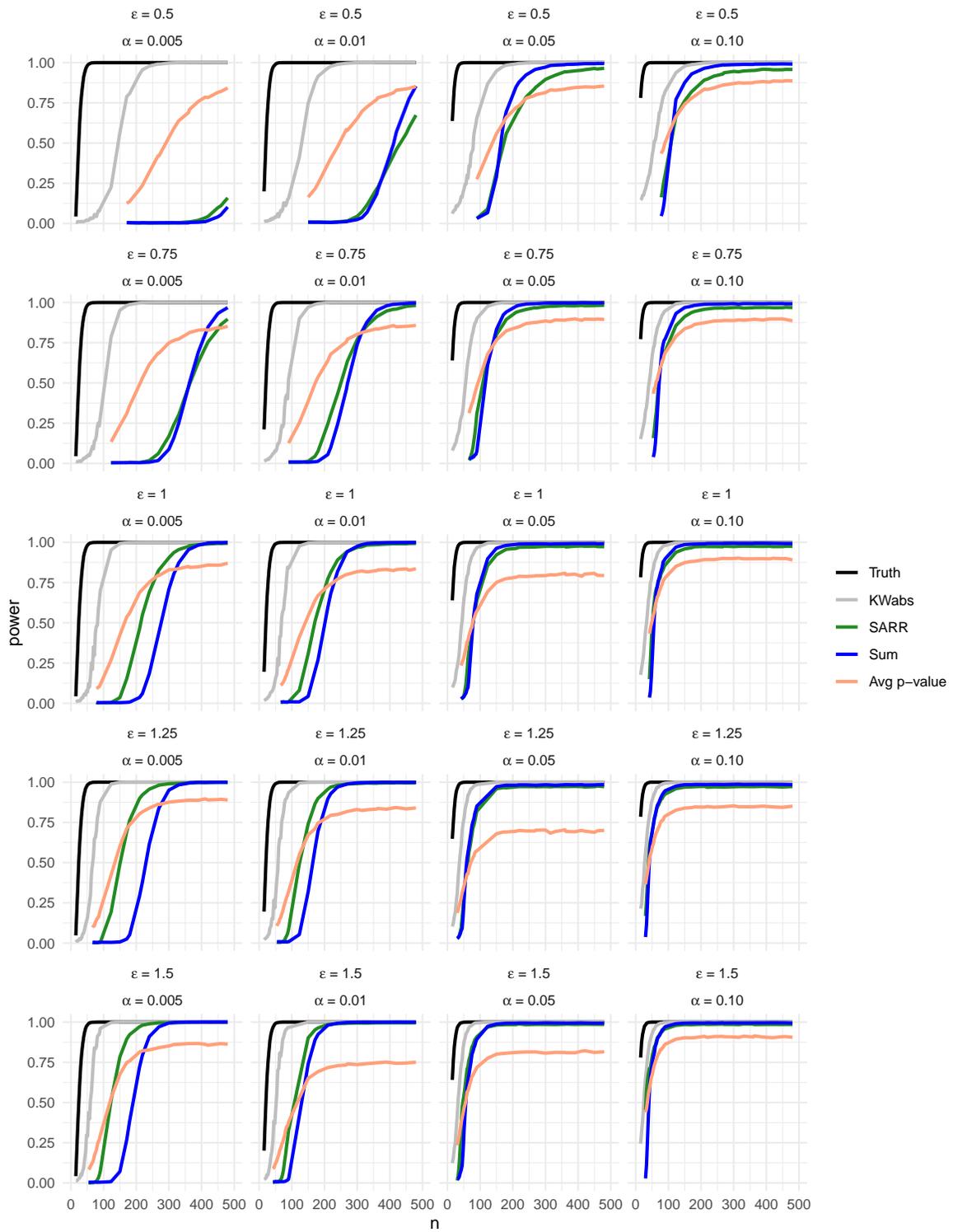}
\caption{Kruskal-Wallis test: Power for different values of $\alpha$ and $\varepsilon$ for a range of total sample sizes $n$.}
\label{fig:kruskal}
\end{figure}

\section{Discussion and future work} \label{sec:discussion}

The subsampled and aggregated randomized response mechanism is a simple and effective tool for constructing differentially private tests from nonprivate tests. In our illustrations, we have seen that the method is especially useful when the type I error rate $\alpha$ is small and $\varepsilon$ is greater or equal to 1. 

We have focused on hypothesis testing, but the subsampled and aggregated randomized response mechanism can be useful in other contexts, especially when the data are naturally split into groups. One such example is federated learning \citep{konevcny2016federated}, where the data are assumed to be stored in different clients. Another application is multiagent decision problems with confidentiality constraints, where the goal is making a collaborative decision ensuring that the individual recommendations are kept private. 

{{A drawback of our approach is that it may not be implemented if the sample size, the type I error $\alpha$ and the privacy parameter $\varepsilon$ are all small (see Table~\ref{tab:minalpha}). However, in those instances, the differentially private tests that can be implemented are low-powered.}}

{{In Section~\ref{subsec:kw}, we compared general-purpose algorithms for differentially private testing to a test proposed in \cite{couch2019differentially} that was specifically developed for the Kruskal-Wallis test. The test proposed in \cite{couch2019differentially} was considerably more powerful than the general-purpose algorithms for small $\alpha$ and $\varepsilon$, but its performance was comparable to that of the subsampled and aggregated randomized response mechanism when $\alpha \ge 0.05$ and $\varepsilon \ge 1$.}}

The subsampled and aggregated randomized response mechanism can be extended in a number of ways. For instance, the mechanism could output a categorical variable with multiple categories instead of a binary decision. Another extension could accommodate multi-step multiple hypothesis testing methods such as the Benjamini-Hochberg procedure \citep{benjamini1995controlling}. In that case, the privacy level of the algorithm should be computed with care because the outputs of the tests become dependent. 

\bibliographystyle{chicago}
\bibliography{main}

\clearpage
\appendix

\begin{center}
{\huge{Supplementary material}}
\end{center}

This document contains supplementary material to the main text of the article. Section~\ref{sec:aux} includes auxiliary results needed to prove the results in the main text, which are proved in Section~\ref{sec:proofs}. In Section~\ref{sec:ttest}, we present an additional simulation study where we compare our method to the differentially private $t$-test proposed in \cite{barrientos2019differentially}. Lastly, Section~\ref{sec:kurtosis} presents the results of the test for the kurtosis of errors proposed in \cite{pena2006global}, in the context of the simulation study reported in Section 4.1 of the main text. 

\newtheorem{prop}{Proposition}[section]

\section{Auxiliary results} \label{sec:aux}

First, we prove auxiliary propositions that are helpful for proving the results in the main text. All of them use theorems in \cite{shaked2007stochastic}. 

We use the notation $\mathrm{Binomial}(i,p) + \mathrm{Binomial}(j,p)$ for the distribution of the sum of independent $\mathrm{Binomial}(i,p)$ and $\mathrm{Binomial}(j,p)$ random variables, with the understanding that if the number of trials is zero, the random variable is zero with probability one.

\begin{definition}
Let $X$ and $Y$ be discrete random variables with common support $S$, which is a subset of the integers. We say that $X$ stochastically dominates $Y$ with respect to the likelihood ratio order if
${\mathbbm{P}(X=t)}/{\mathbbm{P}(Y = t)}$
is increasing in $t$ for $t \in S$. 
\end{definition}

\begin{prop} \label{prop:genlrdom}
Let $B_i \sim \mathrm{Binomial}(n-i, 1-p) + \mathrm{Binomial}(i, p)$ with $1/2 < p < 1$ for $i \in \{0, 1, \, ... , \, n \}$ . For any $\{i , j \} \subset \{0, 1, \, ... \, , n\}$ such that $j > i$, $B_j$ stochastically dominates $B_i$ with respect to the likelihood ratio ordering. 

\end{prop}

\begin{proof}
The result follows by an application of Theorem 1.C.9. in \cite{shaked2007stochastic}. To see this, let $\{i , j \} \subset \{0, 1, \, ... \, , n\}$ such that $j > i$. Let $C_i, C_j \sim \mathrm{Binomial}(i,p)+\mathrm{Binomial}(n-j,1-p)$, $N_i \sim \mathrm{Binomial}(j-i, 1-p)$, and $N_j \sim \mathrm{Binomial}(j-i, p)$. We can write $B_i = C_i + N_i$ and $B_j = C_j + N_j$. Binomial random variables have log-concave probability mass functions, $C_j$ stochastically dominates $C_i$ with respect to the likelihood ratio ordering (they are equal in distribution), and $N_j$ dominates $N_i$ with respect to the likelihood ratio ordering because $p > 1/2$ by assumption. We can apply Theorem 1.C.9. in \cite{shaked2007stochastic} and the result follows.
\end{proof}

\begin{prop} \label{prop:genhrdom}
Let $1/2 < p < 1$ and $B_i \sim \mathrm{Binomial}(i, p) + \mathrm{Binomial}(n-i, 1-p)$ for $i \in \{0, 1, \, ... , \, n \}$. For any $\{i , j \} \subset \{0, 1, \, ... \, , n\}$ such that $j > i$ and $x \in \{0,1, \, .... \, , n\}$,
$$
\frac{\mathbbm{P}(B_i = x) }{\mathbbm{P}(B_i > x)} \ge \frac{\mathbbm{P}(B_j = x)}{\mathbbm{P}(B_j  > x)}.
$$
\end{prop}

\begin{proof}
This is a direct consequence of Proposition~\ref{prop:genlrdom} and the fact that stochastic domination according to the likelihood ratio order implies stochastic domination according to the hazard ratio order (see e.g. Theorem 1.C.1. \cite{shaked2007stochastic}). 
\end{proof}

\begin{prop}  \label{prop:gendecrBi}
Let $1/2 < p < 1$ and  $B_i \sim \mathrm{Binomial}(n-i,1-p) + \mathrm{Binomial}(i, p)$ and $x \in \{0, 1, \, ... \, , n \}$. Then, the ratio
$$
r_j = \frac{\mathbbm{P}(B_i > x)}{\mathbbm{P}(B_{i-1} > x)}
$$
is decreasing in $i$.
\end{prop}

\begin{proof} It suffices to show that for $i \in \{1,2, \, ... \, , n-1 \}$ and $x \in \{0,1, \, ... \, , n\}$, 
$$
\frac{\mathbbm{P}(B_i > x)}{\mathbbm{P}(B_{i-1} >x)}  \ge \frac{\mathbbm{P}(B_{i+1} > x)}{\mathbbm{P}(B_i > x)}.
$$
Let $C_{i+1} \sim \mathrm{Binomial}(i,p)+\mathrm{Binomial}(n-i-1,1-p)$.  Then,
\begin{align*}
\mathbbm{P}(B_{i+1} > x)-\mathbbm{P}(B_i > x) &= (2p-1) \mathbbm{P}(C_{i+1} = x) \\
\mathbbm{P}(B_i > x) &= (1-p) \mathbbm{P}(C_{i+1} > x-1) + p \mathbbm{P}(C_{i+1} > x).
\end{align*}
Similarly, letting $C_i \sim \mathrm{Binomial}(i-1,p) + \mathrm{Binomial}(n-i,1-p)$,
\begin{align*}
\mathbbm{P}(B_i > x) - \mathbbm{P}(B_{i-1} > x) &= (2p-1) \mathbbm{P}(C_i = x) \\ 
\mathbbm{P}(B_{i-1} > x) &= (1-p) \mathbbm{P}(C_i > x-1) + p \mathbbm{P}(C_i > x).
\end{align*}
Now,
\begin{align*}
\frac{\mathbbm{P}(B_i > x)}{\mathbbm{P}(B_{i-1} > x)} \ge \frac{\mathbbm{P}(B_{i+1} > x)}{\mathbbm{P}(B_i > x)} &\Leftrightarrow 
\frac{\mathbbm{P}(B_{i-1} > x)}{\mathbbm{P}(B_i > x) - \mathbbm{P}(B_{i-1} > x)} \le \frac{\mathbbm{P}(B_i > x)}{\mathbbm{P}(B_{i+1} > x)-\mathbbm{P}(B_i > x)}.
\end{align*}
The inequality on the right-hand side is equivalent to
\begin{align*}
(1-p) \frac{\mathbbm{P}(C_i > x-1)}{\mathbbm{P}(C_i = x)} + p \frac{\mathbbm{P}(C_i > x)}{\mathbbm{P}(C_i = x)} \le (1-p) \frac{\mathbbm{P}(C_{i+1} > x-1)}{\mathbbm{P}(C_{i+1}= x)} + p \frac{\mathbbm{P}(C_{i+1} > x)}{\mathbbm{P}(C_{i+1} = x)}
\end{align*}
Since $C_{i+1}$ stochastically dominates $C_{i}$ according to the likelihood ratio ordering, Proposition~\ref{prop:genhrdom} implies that $\mathbbm{P}(C_i > x)/\mathbbm{P}(C_i = x) \le \mathbbm{P}(C_{i+1} > x)/\mathbbm{P}(C_{i+1} = x)$. It remains to show that 
 $$
\frac{ \mathbbm{P}(C_i > x-1)}{\mathbbm{P}(C_{i} = x)} \le \frac{ \mathbbm{P}(C_{i+1} > x-1)}{\mathbbm{P}(C_{i+1} = x)}.
 $$
 The inequality is true after substituting  $\mathbbm{P}(C_j > x -1) = \mathbbm{P}(C_j > x) + \mathbbm{P}(C_j = x)$ in both numerators and then applying Proposition~\ref{prop:genhrdom}.
\end{proof}

\begin{prop}
\label{prop:decrRx}
Let $1/2 < p < 1$ and  $B_i \sim \mathrm{Binomial}(n-i, 1-p) + \mathrm{Binomial}(i, p)$. Then, the ratio
$$
r_x = \frac{\mathbbm{P}(B_i > x)}{\mathbbm{P}(B_{i-1} > x)}
$$
is increasing in $x$ for $x \in \{0,1, \, ... \, , n\}$.
\end{prop}

\begin{proof}
Let $i \in \{0, 1, \, ... , \, n \}$ and $x \in \{0,1, \, ... \, n-1\}$. It is enough to show that $r_x \le r_{x+1}$. Now,
\begin{align*}
r_x \le r_{x+1} &\Leftrightarrow \frac{\mathbbm{P}(B_{i} > x)}{\mathbbm{P}(B_{i} > x+1)} \le   \frac{\mathbbm{P}(B_{i-1} > x)}{\mathbbm{P}(B_{i-1} > x+1)} \\
&\Leftrightarrow \frac{\mathbbm{P}(B_{i} = x+1) + \mathbbm{P}(B_{i} > x + 1)}{\mathbbm{P}(B_{i} > x+1)} \le   \frac{\mathbbm{P}(B_{i-1} = x +1) + \mathbbm{P}(B_{i-1} > x + 1)}{\mathbbm{P}(B_{i-1} > x+1)},
\end{align*}
which is equivalent to 
$$
\frac{\mathbbm{P}(B_i = x+1)}{\mathbbm{P}(B_i > x+1)} \le \frac{\mathbbm{P}(B_{i-1} =x+1)}{\mathbbm{P}(B_{i-1} > x+1)}.
$$
The inequality above is shown to be true in Proposition~\ref{prop:genhrdom}. 
\end{proof}

\section{Proofs of Propositions in main text} \label{sec:proofs}

\newtheorem{prop2}{Proposition}

\begin{prop2}
Let $\varepsilon > 0$ and $p = \exp(\varepsilon)/[1+\exp(\varepsilon)]$. Then, $r(x)$ is exactly $\varepsilon$-differentially private. 

\end{prop2}
\begin{proof}

This is a standard result. It follows directly from the definition of differential privacy.
\end{proof}


\begin{prop2} 
The statistic $d_c = \mathbbm{1}( T > c)$ is exactly $\varepsilon$-differentially private with 
$$
\varepsilon =   \log \left( \frac{\P(B_1 > c_\ast) }{\P(B_0 > c_\ast)} \right),
$$
where $c_\ast = \max(c, 2k-c)$ and  $B_i \sim \mathrm{Binomial}(i,p)+ \mathrm{Binomial}(2k+1-i, 1-p)$ for $i \in \{0,1\}$.
\end{prop2}
\begin{proof}
Let $1/2 < p < 1$ and $B_i \sim \mathrm{Binomial}(i, p) + \mathrm{Binomial}(2k+1-i, 1-p)$ for $i \in \{0, 1, \, ... , \, 2k+1 \}$. Let $r_i = \P(B_i > c)$ be the probability of rejecting the null hypothesis given that $\sum_{j=1}^{2k+1} x_j = i$. By the definition of differential privacy,  
$$
\exp(\varepsilon) = \max_{i \in \{1, \, ... \, , 2k+1 \}} \max \left\{ \frac{r_i}{r_{i-1}}, \frac{r_{i-1}}{r_i}, \frac{1-r_i}{1-r_{i-1}}, \frac{1-r_{i-1}}{1-r_i} \right \}. 
$$
From Proposition~\ref{prop:genlrdom}, we know that $r_i > r_{i-1}$ (first order or ``usual'' stochastic domination is implied by likelihood ratio domination, see e.g. Theorem 1.C.1 in \cite{shaked2007stochastic}). Therefore, 
$$
 \max \left\{ \frac{r_i}{r_{i-1}}, \frac{r_{i-1}}{r_i}, \frac{1-r_i}{1-r_{i-1}}, \frac{1-r_{i-1}}{1-r_i} \right \} = \max \left\{ \frac{r_i}{r_{i-1}}, \frac{1-r_{i-1}}{1-r_i} \right\}.
$$
From Proposition~\ref{prop:gendecrBi}, we know that $r_i/r_{i-1}$ is decreasing in $i$. This narrows down our candidates for the maximum to
$$
\exp(\varepsilon) = \max \{r_1/r_0, (1-r_{2k})/(1- r_{2k+1})\}.
$$
Note that
\begin{align*}
    1-r_{2k+1} &= 1-\P(B_{2k+1} > t) = 1-\P(B_0 < 2k+1 - t) = \P(B_0 > 2k-t)\\
    1-r_{2k} &= 1-\P(B_{2k} > t) = 1-\P(B_1 < 2k+1 - t) = \P(B_1 > 2k-t).
\end{align*}
We can rewrite
$$
\exp(\varepsilon) = \max \left\{ \frac{\P(B_1 > c)}{\P(B_0 > c)}, \frac{\P(B_1 > 2k-c)}{\P(B_0 > 2k-c)} \right\}.
$$
The result follows because Proposition~\ref{prop:decrRx} shows that the ratio $\P(B_1 > x)/\P(B_0 > x)$ is increasing in $x$. 

\end{proof}

\begin{prop2} 
The statistic $d_c = \1(T > c)$ has the following properties:
\begin{enumerate}
    \item For any fixed $k$ and $c$, $\varepsilon$ is increasing in $p$.
    \item For any fixed $p$ and $c \ge k$, $\varepsilon$ is decreasing in $k$. 
    \item For any fixed $k$ and $p$, $\varepsilon$ is minimized at $c = k$. 
\end{enumerate}
\end{prop2}
\begin{proof} 

First, we show that for any fixed $k$ and $c$, $\varepsilon$ is increasing in $p$. We do so by showing that $1/(\exp(\varepsilon) -1)$ is decreasing in $p$. We can write
$$
\frac{1}{\exp(\varepsilon)-1} = \frac{\P(B_0 > c_\ast)}{\P(B_1 > c_\ast) - \P(B_0 > c_\ast)}.
$$
Let $B \sim \mathrm{Binomial}(2k, 1-p)$. Then, 
$$
\P(B_1 > c_\ast) - \P(B_0 > c_\ast) = (2p-1)\P(B= c_\ast).
$$
Plugging in our new expression for $\P(B_1 > c_\ast) - \P(B_0 > c_\ast)$, we obtain
$$
\frac{1}{\exp(\varepsilon)-1} = \frac{\P(B_0 >c_\ast)}{(2p-1) \P(B=c_\ast)} = \frac{1-p}{2p-1} \frac{P(B>c_\ast-1)}{P(B= c_\ast)} + \frac{p}{2p-1} \frac{P(B > c_\ast)}{P(B = c_\ast)}.
$$
Rearranging terms,
$$
\frac{1}{\exp(\varepsilon)-1} = \frac{1-p}{2p-1} + \frac{1}{2p-1} \frac{P(B>c_\ast)}{P(B=c_\ast)}.
$$
The result follows because $(1-p)/(2p-1)$ and $1/(2p-1)$ are decreasing in $p$ for $1/2 < p < 1$, and $P(B > c_\ast)/P(B = c_\ast)$ is also decreasing in $p$. The  latter is true because if $p$ increases, $B$ is stochastically decreasing according to the likelihood ratio and hazard ratio order (see e.g. Example 1.C.51. in \cite{shaked2007stochastic}), which in turn implies that $P(B > c_\ast)/P(B = c_\ast)$ is decreasing, as required.

Now, we show that for any fixed $p$ and $c \ge k$, $\varepsilon$ is decreasing in $k$. Since $c \ge k$, we know by Proposition~\ref{prop:epsilon} that $c_\ast = c$ and $
\exp(\varepsilon) = {\P(B_1 > c)}/{\P(B_0 > c)}.$ Let $B \sim \mathrm{Binomial}(2k, 1-p)$. Then, the ratio can be rewritten as
$$
 \frac{\mathbbm{P}(B_1 > c)}{\mathbbm{P}(B_0 > c)}  = 1 +  \frac{2p-1}{1-p + \frac{\mathbbm{P}(B > c)}{\mathbbm{P}(B = c)}}
$$
Therefore, $\varepsilon$ is decreasing in $k$ if and only if $\P(B=c)/\P(B>c)$ is decreasing in $k$. The result follows because $B$ is stochastically increasing in $k$ with respect to the likelihood ratio order, so the ratio $\P(B=c)/\P(B>c)$ is decreasing in $k$ (again, this is implied by the fact that stochastic domination with respect to the likelihood ratio order implies domination with respect to the hazard ratio order).

Finally, we show that for any fixed $k$ and $p$, $\varepsilon$ is minimized at $c = k$. From Proposition~\ref{prop:decrRx}, we know that $\varepsilon$ is increasing in $c_\ast$. Since $c_\ast = \max(c, 2k-c)$, it follows that for fixed $k$, and $p$, $\varepsilon$ attains its minimum at $c = k$, as required. 

\end{proof}

\begin{prop2} 
The statistic $d = \1(T > k)$ has the following properties:
\begin{enumerate}
    \item For any fixed $p$, 
$$
\lim_{k \rightarrow \infty} \varepsilon  = \log\left(1+\frac{(2p-1)^2}{2p(1-p)}\right) > 0.
$$
\item A necessary condition on $p$ for achieving $\varepsilon$ differential privacy is
$$
p \le \frac{1}{2} \left( 1 + \frac{\sqrt{\exp(2 \varepsilon)-1}}{1+\exp(\varepsilon)} \right).
$$
A sufficient condition on $p$ for achieving $\varepsilon$ differential privacy is 
$$
p \le \frac{\exp(\varepsilon)}{1+\exp(\varepsilon)}.
$$
\end{enumerate}

\end{prop2}
\begin{proof}

For proving the results in this proposition, it will be useful to rewrite $\exp(\varepsilon)$ in terms of $B \sim \mathrm{Binomial}(2k, 1-p) $. 

Let $B_0 \sim \mathrm{Binomial}(2k+1, 1-p)$, and $B_1 \sim\mathrm{Binomial}(1, p) + \mathrm{Binomial}(2k, 1-p)  $. By Proposition~\ref{prop:epsilon}, we know that 
$$
\exp (\varepsilon) = \frac{\P(B_1 > k)}{\P(B_0 > k)}.
$$
Let $B \sim \mathrm{Binomial}(2k, 1-p) $. The ratio can be rewritten as
\begin{equation} \label{eq:ratio_b}
 \frac{\mathbbm{P}(B_1 > k)}{\mathbbm{P}(B_0 > k)}  = 1 +  \frac{2p-1}{1-p + \frac{\mathbbm{P}(B > k)}{\mathbbm{P}(B = k)}}.
\end{equation}
This shows that $\exp(\varepsilon)$ depends on $k$ only through 
${\P(B > k)}/{\P(B = k)},$ which can be expressed as
\begin{equation} \label{eq:ratio}
\frac{\mathbbm{P}(B>k)}{\mathbbm{P}(B = k)} = \sum_{i =1}^k \frac{P(B = k + i)}{P(B = k)} = \sum_{i=1}^k \left\{ \prod_{j=1}^i \frac{k-j+1}{k+j} \right\} \left( \frac{1-p}{p} \right)^{i}.
\end{equation}
First, we find the the limit of $\varepsilon$ as $k$ grows for fixed $1/2 < p < 1$.

The product term in Equation~\eqref{eq:ratio} can be bounded as follows: 
$$
\left(\frac{k-i+1}{k+i} \right)^i \le \prod_{j=1}^i \frac{k-j+1}{k+j}  \le 1.
$$
Then,
$$
\lim_{k \rightarrow \infty} \frac{\mathbbm{P}(B > k)}{ \mathbbm{P}(B = k) } \le  \sum_{i=1}^\infty \left( \frac{p}{1-p} \right)^i =  \frac{1-p}{2p-1}.
$$
Consider the series
$$
\lim_{k \rightarrow \infty} \sum_{i=1}^k \left(\frac{k-i+1}{k+i} \right)^i \left( \frac{1-p}{p} \right)^i.
$$
All the terms are positive and the summand is increasing in $k$, so we can apply the monotone convergence theorem for series:
\begin{align*}
\lim_{k \rightarrow \infty} \sum_{i=1}^k \left(\frac{k-i+1}{k+i} \right)^i \left( \frac{1-p}{p} \right)^i = \lim_{k \rightarrow \infty}  \sum_{i=1}^k \left( \frac{1-p}{p} \right)^i = \frac{1-p}{2p-1}.
\end{align*}
Therefore, we conclude that 
$$
\lim_{k \rightarrow \infty} \frac{\mathbbm{P}(B>k)}{\mathbbm{P}(B=k)} = \frac{1-p}{2p-1}.
$$
Plugging the limit into Equation~\eqref{eq:ratio_b}, we obtain:
$$
\frac{\mathbbm{P}(B_1 > k)}{\mathbbm{P}(B_0 > k)} = 1+\frac{(2p-1)^2}{2p(1-p)},
$$
as required.

The second part of the proposition is a direct consequence of previous results. The sufficient condition corresponds to the case $k = 0$, and it is sufficient because $\varepsilon$ is decreasing in $k$, all else being equal.  The necessary condition can be found by solving for $p$ in the limiting expression of $\varepsilon$ when $k$ grows to infinity.

\end{proof}

\begin{prop2} 
For any given $s \in \{0, 1, \, ... \, , k\}$, 
$$\P(d \neq \tilde{d} \mid \textstyle \sum_{i=1}^{2k+1} x_i =  s) = \P(d \neq \tilde{d} \mid \sum_{i=1}^{2k+1} x_i = 2k+1-s).$$ 
\end{prop2}
\begin{proof}
The result can be proved by letting $s \in \{0, 1, \, ... \, , k\}$ and noting that $\P(B_s > k) = P(B_{2k+1-s} \le k)$ for $B_s \sim \mathrm{Binomial}(s,p) + \mathrm{Binomial}(2k+1-s, 1-p)$ and $B_{2k+1-s} \sim \mathrm{Binomial}(2k+1-s, p) + \mathrm{Binomial}(s, 1-p)$.

\end{proof}

\begin{prop2} 
The probability $\mathbbm{P}( d \neq \tilde{d} \mid \textstyle \sum_{i=1}^{2k+1} x_i = s )$ has the following properties:  
\begin{enumerate}
    \item For any fixed $k$ and $p$, $ \mathbbm{P}( d \neq \tilde{d} \mid \textstyle \sum_{i=1}^{2k+1} x_i = s )  $ is  decreasing in $s$ if $s > k$ and increasing in $s$ if $s \le k$.
    \item For any fixed $p$ and $s$, $\mathbbm{P}( d \neq \tilde{d} \mid \textstyle \sum_{i=1}^{2k+1} x_i = s ) $ is {decreasing} in $k$ if $s \le k$ and increasing in $k$ if $s > k$.
    \item For any fixed $k$ and $s$,  $\mathbbm{P}( d \neq \tilde{d} \mid \textstyle \sum_{i=1}^{2k+1} x_i = s )  = 1/2$ if $p = 1/2$ and $\mathbbm{P}( d \neq \tilde{d} \mid \textstyle \sum_{i=1}^{2k+1} x_i = s )  = 0$ if $p = 1$.
\end{enumerate}

\end{prop2}
\begin{proof} 
We prove the three statements separately.

1. If $s > k$, the probability of disagreement is $\P(B_s \le k)$, whereas if $s \le k$, it is $\P(B_s > k)$.  The result follows because, by Proposition~\ref{prop:genlrdom}, $B_s$ is stochastically increasing in $s$.

2. If $s \le k$, then $\mathbbm{P}(d \neq \tilde{d}) = \P(B_s > k)$, which is decreasing in $k$. To see this, let $k$ be fixed and increase it by one, defining $k_\ast = k+1$. Then, we can define $B^\ast_s \sim B_s + \mathrm{Bernoulli}(1-p)$. 
$$
\P(B^\ast_s > k_\ast) = (1-p) \P(B_s > k) + p \P(B_s > k+1),
$$
which is smaller than $\P(B_s > k)$ because $p \ge 1/2$ and $\P(B_s > k +1) < \P(B_s > k)$. If $s > k$, then $\P(d \neq d) = \P(B_{s} \le k)$, and a similar argument to the one we just used shows that it is increasing in $k$.

3. This proof of this part is direct given the expression of $\mathbbm{P}( d \neq \tilde{d} \mid \textstyle \sum_{i=1}^{2k+1} x_i = s )$.

\end{proof}

\begin{prop2} 
The probability that $d$ rejects $H_0$ has the following properties:  
\begin{enumerate}
    \item For any fixed $k$ and $p$,  the probability that $d$ rejects $H_0$ is decreasing in $\gamma_0$.
    \item For any fixed $\gamma_0$ and $k$, the probability that $d$ rejects $H_0$ is decreasing in $p$ if $\gamma_0 < 1/2$ and increasing in $p$ if $\gamma_0 \ge 1/2$.
    \item Let $p > 1/2$ be fixed. If $\gamma_0 > 1/2$, then the probability that $d$ rejects $H_0$ goes to 1 as $k \rightarrow \infty$. Alternatively, if $\gamma_0 < 1/2$, then the probability that $d$ rejects $H_0$ goes to 0 as $k \rightarrow \infty$.
\end{enumerate}

\end{prop2}
\begin{proof} 
The probability that $d$ rejects $H_0$ is $\P(T > k)$ for $T \sim \mathrm{Binomial}(2k+1, p \gamma_0 + (1-p) (1-\gamma_0))$. We prove the three statements separately.

1. Since $p \ge 1/2$,  $\P(T > k)$ is increasing in $\gamma_0$.

2. This is similar to part 1. If $\gamma_0 < 1/2$, then $p \gamma_0 + (1-p)(1-\gamma_0)$ is decreasing in $p$. If $\gamma_0 \ge 1/2$, then it is increasing in $p$. 

3. If $p > 1/2$ and $\gamma_0 > 1/2$, then $p \gamma_0 + (1-p)(1-\gamma_0) > 1/2$ and $\P(T > k)$ goes to one as $k$ goes to infinity. Similarly, if $p > 1/2$ and $\gamma_0 < 1/2$, then  $p \gamma_0 + (1-p)(1-\gamma_0) < 1/2$ and $\P(T > k)$ goes to zero as $k$ goes to infinity. One can find these limits using standard tail bounds for the binomial distribution.

\end{proof}

\begin{prop2} 
For any $\varepsilon > 0$, the minimum type I error $\alpha$ attainable by $d$ goes to zero as $k$ goes to infinity.
\end{prop2}
\begin{proof} 
The minimum type I error is achieved when $\alpha_0 = 0$. The type I error of $d$ is then $\P(T > k)$ for $T \sim \mathrm{Binomial}(2k+1,  1-p)$, where $p > 1/2$ is set given $\varepsilon$ and $k$. The probability goes to zero as $k$ goes to infinity because $p > 1/2$.
\end{proof}

{{\section{Differentially private $t$-test \label{sec:ttest} }}

{{In this section, we report the result of a simulation study where we compare the performance of the subsampled-and-aggregated randomized response mechanism to the differentially private $t$-test for regression proposed in \cite{barrientos2019differentially}. This is a specific test that is only applicable to this task.}}

{{The data are simulated from the normal linear model $y = X \beta + \varepsilon$ where $\beta = [\beta_1, \beta_2, \beta_3, \beta_4, \beta_5]' = [1, 1, 1, 1, 1]'.$ We test $H_0: \beta_1 = 0$ against $H_1 = \beta_1 \neq 0$. For the method in \cite{barrientos2019differentially}, we set the number of subsets to 25 and the truncation parameter to $a = 2$, following what was proposed in \cite{barrientos2019differentially}. The number of subgroups for the other methods based on data-splitting is set using the strategy proposed in Section~\ref{sec:tuning} with $\alpha_{0, \min} = \alpha$. We consider $\alpha \in \{0.005, 0.1, 0.05, 0.1\}$, $\varepsilon \in \{0.5, 0.75, 1, 1.25, 1.5\}$ and a range of sample sizes $n$ that runs up to $10^5$. For each scenario, we perform $10^4$ simulations.}}

{{The test proposed in \cite{barrientos2019differentially} (labeled as DP t) outperforms the general-purpose algorithms in all cases except when $\varepsilon = 0.5$. The binarized sum and the average $p$-value are best in this case. The randomized response mechanism performs best when $\alpha = 0.005$ and $\varepsilon \in \{1.25, 1.5\}$.}}

\begin{figure}
\includegraphics[width=\linewidth]{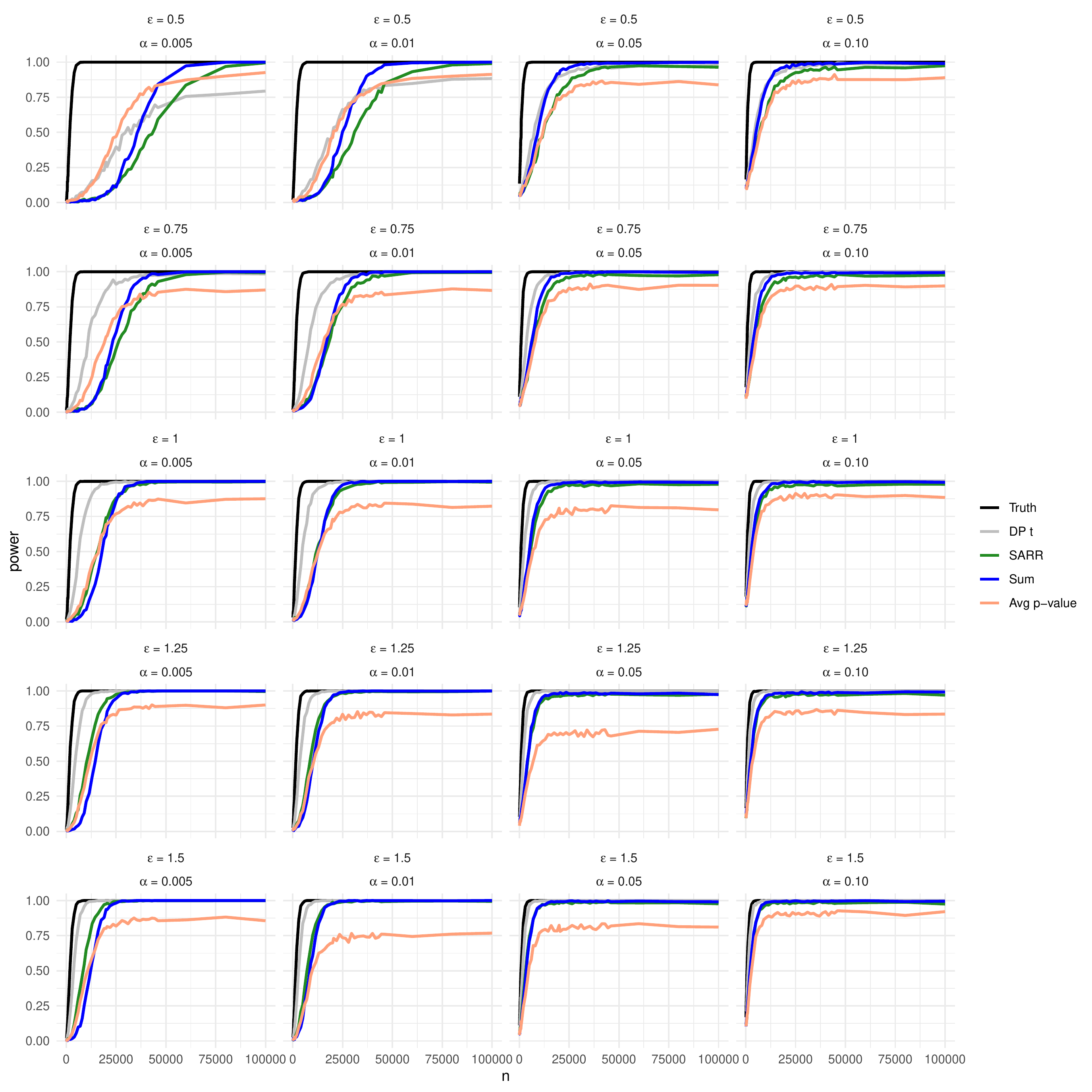}
\caption{Differentially private $t$-test: Average power of $t$-tests for regression for different combinations of $\alpha$ and $\varepsilon$ as a function of the total sample size $n$.}
\label{fig:ttest}
\end{figure}

{{\section{Goodness-of-fit: Kurtosis} \label{sec:kurtosis}}}

{{Figure~\ref{fig:kurtosis} displays the result of the test for kurtosis in Section 4.1 of the main text. Like we observed in the other scenarios, randomized response performs best when $\alpha = 0.005$ and $\varepsilon$ is greater or equal to 1. The sum performs best for high $\alpha$ and low $\varepsilon$. The average $p$-value is best when both $\alpha$ and $\varepsilon$ are small.}}

\begin{figure}
\includegraphics[width=\linewidth]{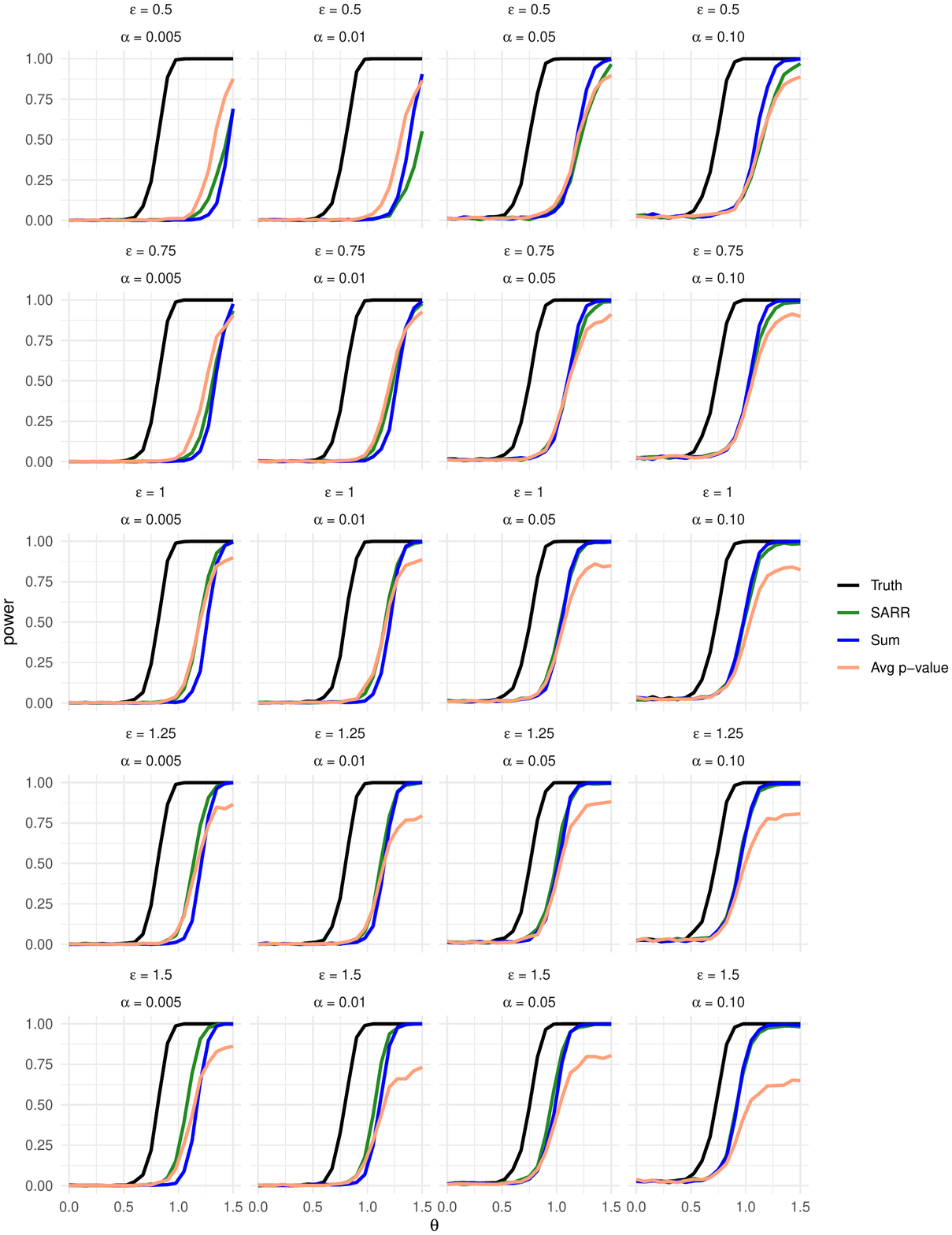}
\caption{Goodness-of-fit tests: Average power of tests for kurtosis for different combinations of $\alpha$ and $\varepsilon$.}
\label{fig:kurtosis}
\end{figure}

\end{document}